\documentclass[10pt,reqno]{amsart}
 \usepackage[foot]{amsaddr}
\usepackage{amsmath}
\usepackage{latexsym}
\usepackage{amsfonts}
\usepackage{amssym b}
\usepackage{color}
\usepackage{bbm,dsfont}
\usepackage{graphicx}
\usepackage{hyperref}

\newtheorem{proposition}{Proposition}
\newtheorem{proposition?}{Proposition?}
\newtheorem{theorem}{Theorem}
\newtheorem{lemma}{Lemma}
\newtheorem{corollary}{Corollary}

\newtheorem{definition}{Definition}

\newcommand{\hi}{\mathcal{H}} 
\newcommand{\id}{\mathbbm{1}} 

\newcommand{\manif}{\mathcal{D}}

\newcommand{\dm}{\dot{\mathsf D}}

\usepackage[left=2cm,top=2cm,right=2cm, bottom=1.95cm, a4paper]{geometry}

\begin{document}
\title[]{Information geometry and local asymptotic normality for multi-parameter estimation of quantum Markov dynamics}

\author[Guta]{Madalin Guta}
\email{madalin.guta@nottingham.ac.uk}
\address{School of Mathematical Sciences, University of Nottingham, University Park,
Nottingham, NG7 2RD, UK}

\author[Kiukas]{Jukka Kiukas}
\email{jek20@aber.ac.uk}

\address{Department of Mathematics, Aberystwyth University, Penglais, Aberystwyth, Ceredigion, SY23 3BZ, UK}

\date{}

\begin{abstract} {This paper deals with the problem of identifying and estimating dynamical parameters of continuous-time quantum open systems, in the input-output formalism. 
First, we characterise the space of identifiable parameters for ergodic dynamics, assuming full access to the output state for arbitrarily long times, and show that the equivalence classes of undistinguishable parameters are orbits of a Lie group acting on the space of dynamical parameters. 
Second, we define an information geometric structure on this space, including a principal bundle given by the action of the group, as well as a compatible connection, and a Riemannian metric based on the quantum Fisher information of the output. We compute the metric explicitly in terms of the Markov covariance of certain "fluctuation operators", and relate it to the horizontal bundle of the connection. 
Third, we show that the system-output and reduced output state satisfy local asymptotic normality, i.e. they can be approximated by a Gaussian model consisting of coherent states of a multimode continuos variables system constructed from the Markov covariance ``data". We illustrate the result by working out the details of the information geometry of a physically relevant two-level system.
}
\end{abstract}


\maketitle

\section{Introduction}

The input-output formalism \cite{GardinerZoller,Breuer&Petruccione} is fundamental to key areas of quantum open systems theory such as Markov dynamics, continuous-time measurements and filtering theory \cite{Bel92b,Bouten&vanHandel&James}, quantum networks \cite{Gough&James} and feedback control \cite{Bel99,Nurdin&James&Petersen}. The formalism serves as a platform which integrates in  a common language methods from control engineering, classical and quantum stochastic processes, non-equilibrium statistical mechanics, and quantum information.  In this paper we aim to further expand this platform by adopting a system identification  \cite{Ljung} perspective. Concretely, we investigate which dynamical parameters of an open system can be estimated from the output state (identifiability problem), how the associated quantum Fisher information arises from the structure of the parameter manifold (information geometry), and how the multi-parameter statistical model defined by the output state can be approximated by a quantum Gaussian model (local asymptotic normality).

\begin{figure}[h]
\includegraphics[width=10cm]{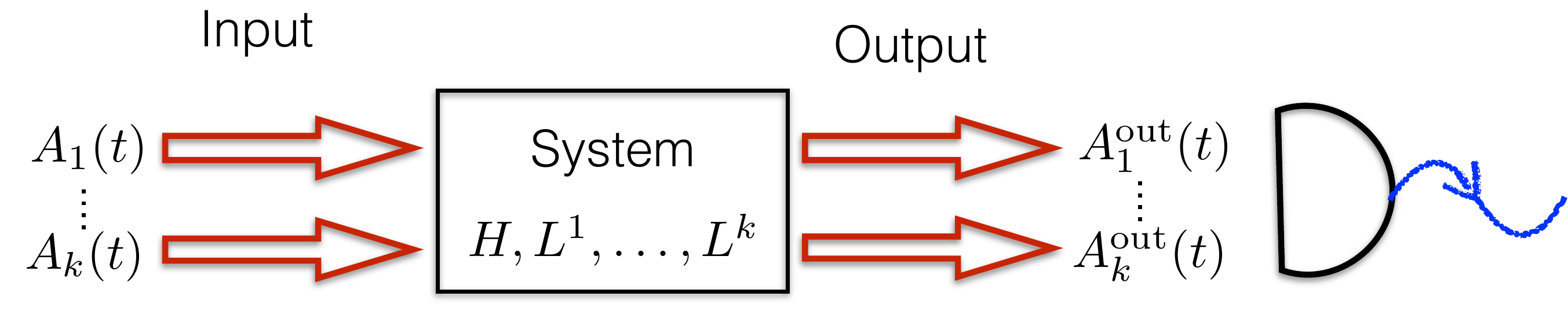}
\caption{Continuous-time Markovian dynamics of an open quantum system in the input-output formalism. 
Input fields $A_i(t)$ interact with the system, so that the joint unitary transformation $U_\mathsf{D}(t)$ depends on the dynamical parameter $\mathsf{D}:= (H,L^1,\dots,  L^k)$ where $H$ is the system Hamiltonian and $L^i$ are the coupling (jump) operators with the input fields. 
The output state carry information about $\mathsf{D}$, which can be estimated by measuring the output fields. 
}\label{fig.markov}
\end{figure}

In a typical quantum input-output set-up, an open system (e.g. an atom, or a cavity mode) is driven by an input consisting of the vacuum or coherent state of the electromagnetic field, the latter being modelled by a continuum of Bosonic modes representing the incoming ``quantum noise", see Figure \ref{fig.markov}. The input interacts with the system in a Markovian fashion, with joint unitary evolution $U_{\mathsf{D}}(t)$ determined by the ``dynamical parameters" $\mathsf{D}:= (H, L^1, \dots,  L^k)$, where $H$ is the system Hamiltonian and $L^i$ is the coupling operator to the $i$-th input mode. 
The output fields carry information about the dynamical parameter $\mathsf{D}$, and can be monitored by means of continuous-time measurements, or may be ``post-processed", e.g. by using feed-forward or feedback schemes \cite{Wiseman&Milburn}. 
However, since such schemes often rely on the knowledge of the dynamics, it is important to develop efficient methods 
for estimating the unknown parameters entering the dynamics. Our goal here is not to propose or analyse specific measurement and estimation schemes (see e.g. \cite{Mab96,Wiseman&Gambetta,Molmer&Gammelmark,CatanaGutaKypraios} for related results), but rather investigate the statistical properties of the output state, which will provide the ultimate limits in estimation precision. We envisage that the structure of the output state uncovered here will be relevant not only for designing efficient measurement schemes (cf. \cite{Guta&Janssen&Kahn} for optimal estimation of qubit states) but also for applications in quantum metrology \cite{MGGL} and quantum control, including feedback.

In our analysis we assume that the system is finite dimensional, and the input is stationary (time independent). We also assume that the dynamics is ergodic, i.e. the system has a unique strictly positive stationary state 
$\rho_{ss}$, in which case any initial state converges to $\rho_{ss}$ and the output becomes stationary in time. From a quantum information perspective, the system-output state $|\Psi^{\rm s+o}(t)\rangle$ associated to the time interval $[0, t]$ is a continuous matrix product state \cite{Verstraete&Cirac}, and the output state 
$\rho^{\rm out}(t)$  is a continuous version of a purely generated finitely correlated state \cite{FNW92}. 
Our results are therefore relevant for the problem of estimating such states, whose discrete version was considered in \cite{BaGr13} from the perspective of quantum tomography of spin chains.

Since we deal with a multi-parameter statistical problem, we adopt a differential geometry viewpoint in the spirit of the theory of information geometry \cite{Amari&Nagaoka}. This allows us to characterise the manifold of identifiable parameters as the quotient of the parameter space with respect to a group of transformations leaving the output state invariant (see Theorem \ref{equiv_thm}), thus extending our previous results for discrete time quantum Markov chains \cite{GutaKiukas}.  An analogous differential geometric construction has been presented in \cite{HaMa12, HaOsVe13} for parametrisations of discrete matrix product states, and a related approach has been used in studying the manifold of correlation matrices for stationary states of certain specific open quantum systems \cite{BaGiZa13}.

Furthermore, we show that the quantum Fisher information (QFI) \cite{Holevo,Braunstein&Caves} of the output is closely related to the covariance of certain ``fluctuation operators", which we study in detail in section \ref{sec.info.geom}. 
The covariance defines a Riemannian metric on the space of identifiable parameters, and provides a complex structure and a positive inner product on the tangent space of identifiable parameters. 
An alternative approach to computing the quantum Fisher information is described in 
\cite{molmer}, see also \cite{MGGL} and \cite{Cozzini&Ionicioiu&Zanardi}.

With the help of this differential geometric structure we construct an associated algebra canonical commutation relations (CCR), 
and a family of coherent states whose QFI is equal to the QFI per time unit of the output state. 
The latter will play the role of limit Gaussian model below.

Local asymptotic normality (LAN) is a key concept in asymptotic statistics, that describes how certain 
statistical models can be approximated by  simpler Gaussian models, with vanishing error in the limit of large ``sample size".  This phenomenon occurs for instance in the case of models consisting of independent, identically distributed samples \cite{LeCam}, but also for multiple observations from an ergodic Markov process \cite{HopfnerJacodLadelli}, or hidden Markov process \cite{BickelRitovRyden2}. In quantum statistics, the general theory of convergence of models was discussed in \cite{GutaJencova,GillGuta}, and LAN for ensembles independent finite dimensional systems was established in \cite{KahnGuta}. For quantum Markov dynamics, LAN for one-dimensional  parameter models was discussed in \cite{Guta11,GutaKiukas} for discrete time, and in \cite{CatanaGutaBouten} for continuous-time. Here we extend the latter to the multi-dimensional model where all identifiable parameters are assumed to be unknown; this brings forward the information-geometric aspects, which do not play a significant role in a one-parameter setting. Theorem \ref{th.lan} shows that the system-output state and (reduced) output state models converge to the Gaussian model consisting of a family of coherent states of the above mentioned CCR algebra, in the limit of large times. 

The present investigation suggests several interesting future lines of research.  One direction is to understand the physical significance of the geodesic distance of the Fisher metric and the relation to quantum speed limit \cite{Taddei&Escher&Davidovich&MatosFilho} and thermodynamic metrics \cite{Sivak&Crooks}. Another direction is to show that fluctuation operators satisfy the Central Limit Theorem, and identify the measurement which achieves the optimal estimation precision. 
Building on \cite{Guta&Yamamoto}, one can develop a similar theory for the identification of quantum linear input-output  
systems in the stationary regime, i.e. from the ``power spectrum".
Moreover, the extension of the current theory to non-ergodic dynamics and the analysis of ``metastable" \cite{MGGL2} or ``near phase transition" \cite{MGGL} systems is important due to its relevance for quantum metrology. Finally, our framework has a number of interesting generalisations connected with other ongoing mathematical work on quantum stochastic evolutions. In particular, when the stationary state manifold is nontrivial (non-ergodic case), one can discuss conserved quantities and adiabatic transport \cite{AvFrGr12,GoRaSm15,Albert}. From the more technical point of view, our manifold of dynamical parameters actually has a natural Lie group structure \cite{EvGoJa12}; reformulation of our results in this more structured framework could be useful especially for applications to control theory.

In order to increase the accessibility of the paper, we collect the main constructions and results in the next section.

\section{Overview of results}\label{sec.overview}

Section \ref{sec.Markov.dynamics} introduces the input-output formalism of quantum open dynamics, as illustrated in Figure \ref{fig.markov}. For a given dynamical parameter 
$\mathsf{D}:= (H, L^1, \dots,  L^k)$, the system-output state is given by $|\Psi_\mathsf{D}^{\rm s+o}(t)\rangle = U_{\mathsf{D}}(t)|\varphi\otimes \Omega\rangle$ where $|\varphi\rangle$ is the initial system state, 
$|\Omega\rangle$ is the input state (taken to be the vacuum), and $U_{\mathsf{D}}(t)$ is the joint unitary evolution given by the quantum stochastic differential equation 
$$
dU_{\mathsf{D}}(t) = 
\left(\sum_{i=1}^{k}(iH \otimes \id_\mathcal{F} \, dt + L^i \otimes dA_i^*(t)-L^{i*}\otimes dA_i(t))-\frac{1}{2}\sum_{i=1}^{k}L^{i*}L^i \otimes \id_\mathcal{F} \, dt\right)U_{\mathsf{D}}(t).
$$
Above, $dA_i(t)$ and $dA^*_i(t)$ are the time increments of input annihilation and creation operators of $k$ Bosonic input channels, acting on the Fock space $\mathcal{F}$ over $L^2(\mathbb{R}_+)\otimes \mathbb{C}^k$. 
The reduced system evolution is governed by an ergodic Markov semigroup with Lindblad generator $\mathbb{W}_\mathsf{D}$, and unique stationary state $\rho^{ss}_\mathsf{D}$. The output state after time $t$ is obtained by tracing out the system, $\rho^{\rm out}_\mathsf{D}(t)= {\rm tr}_s( |\Psi_\mathsf{D}^{\rm s+o}(t)\rangle \langle \Psi_\mathsf{D}^{\rm s+o}(t)|)$. For long times the system converges to the stationary state, and the output becomes stationary in time.

\begin{figure}[h]
\includegraphics[width=6cm]{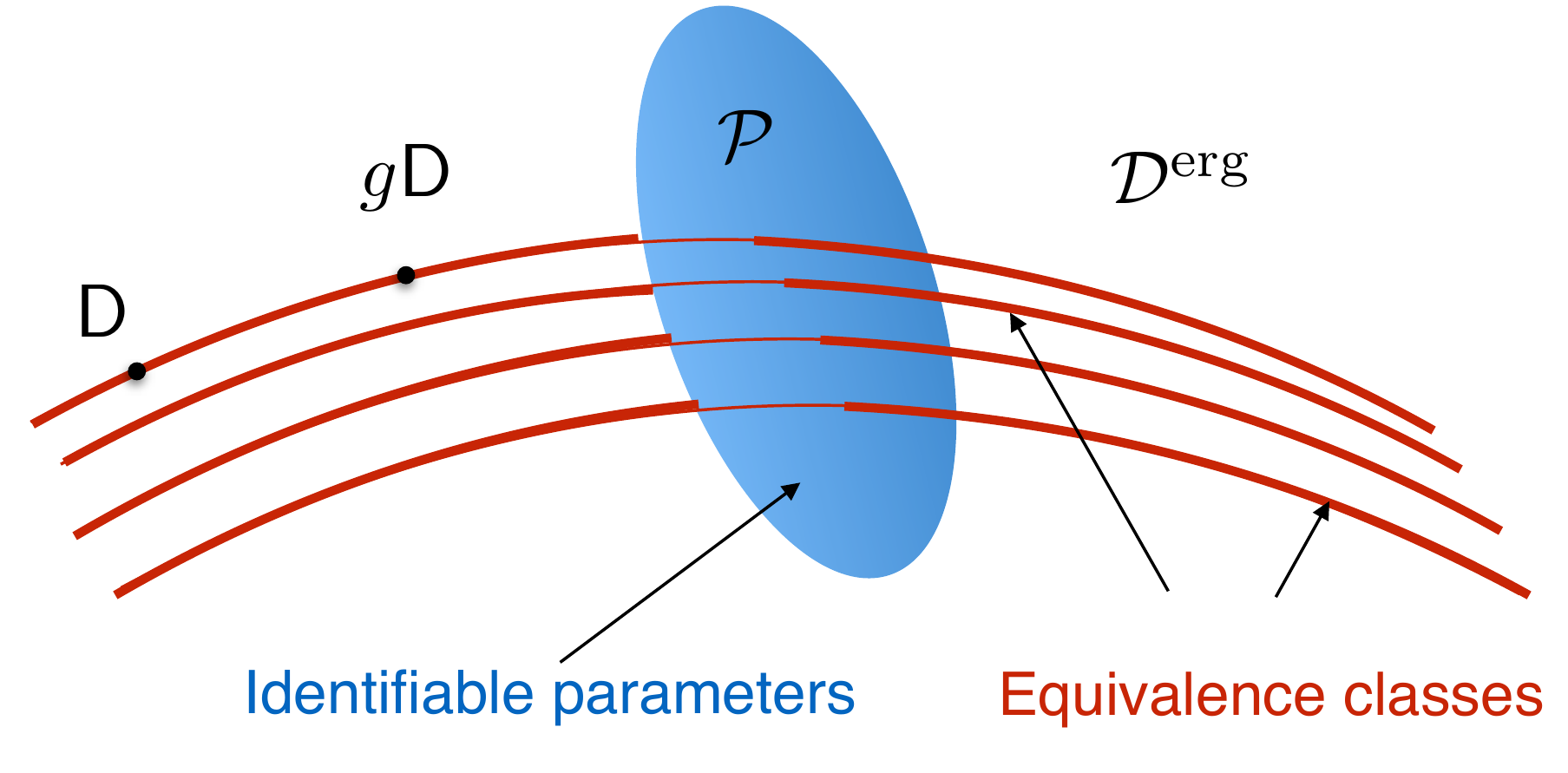}\hspace{2cm}
\includegraphics[width=6cm]{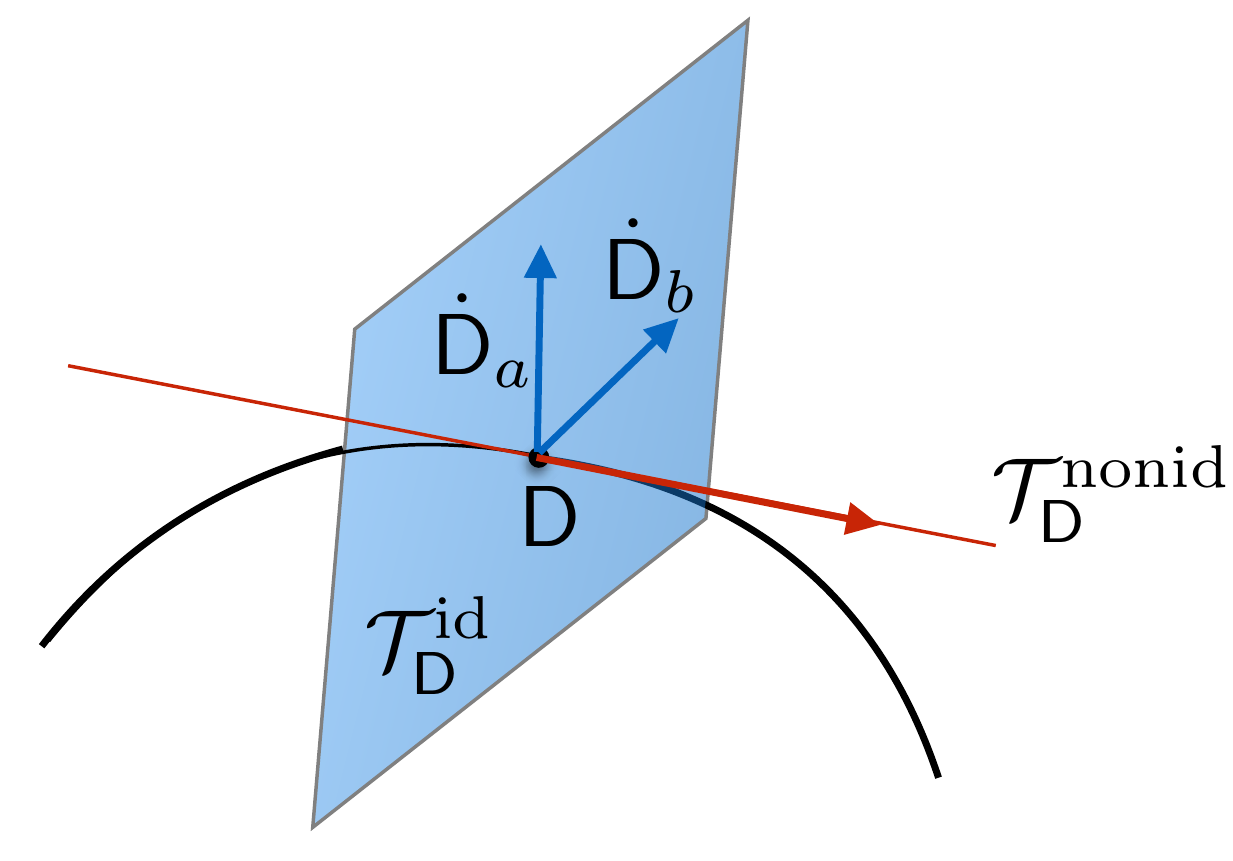}
\caption{Left panel: the space of ergodic dynamical parameters $\mathcal{D}^{\rm erg}$ as principle $G$-bundle over the base manifold $\mathcal{P}$ of identifiable parameters. 
Equivalence classes (red lines) of dynamical parameters with identical outputs. 
 Right panel: the tangent space at the point $\mathsf{D}$ decomposes as direct sum of the 
 tangent space $\mathcal{T}^{\rm nonid}_\mathsf{D}$ to the orbit of the group action, and the space 
 $\mathcal{T}^{\rm id}_\mathsf{D}$ of ``identifiable directions" defined by the identity $\mathcal E^0_{\mathsf{D}}(\dm)=0$. The Markov covariance defines a complex structure and an inner product on $\mathcal{T}^{\rm id}_\mathsf{D}$, such that the QFI rate is $f_{a,b} =4 {\rm Re}(\dot{\mathsf{D}}_a ,\dot{\mathsf{D}}_a )_\mathsf{D}$.}
 \label{fig.identifiability.classes}
\end{figure}

Section \ref{sec.identifiability} discusses the identifiability problem in the stationary setting, see Figure \ref{fig.identifiability.classes}. 
We define an equivalence relation  between dynamical parameters for which the \emph{stationary output states} are identical for all times $t$. 
In Theorem \ref{equiv_thm} we show that two dynamical parameters $\mathsf{D}$ and $\mathsf{D}^\prime$ are equivalent \emph{if and only if} they are related by the ``gauge transformation" $H^\prime = W^*HW+  r\id$ and $ L^\prime_i = W^* L_i W$, where $W$ is unitary and $r$ is a real constant. From a differential geometry viewpoint, the space of \emph{identifiable parameters} is the quotient 
$\mathcal{P}:=\mathcal{D}^{erg}/G$, where $\mathcal{D}^{erg}$ is the manifold of dynamical parameters $\mathsf{D}$ with ergodic dynamics, and $G= PU(d)\times \mathbb{R}$ is the group of ``gauge transformations" whose orbits are the equivalence classes of parameters.
In particular, we show that $\mathcal{D}^{erg}$ is a principal $G$-bundle over the base manifold $\mathcal{P}$. 
The vertical bundle over $\mathcal{D}^{erg}$ consists of subspaces $\mathcal{T}^{\rm nonid}_{\mathsf{D}}$ of the tangent space 
$\mathcal{T}_{\mathsf{D}}$ at $\mathsf{D}$, corresponding to un-identifiable changes of parameters, i.e. infinitesimal changes induced by the action of the group 
$G$. Although in general there is no canonical decomposition of the tangent space into ``identifiable" and ``non-indentifiable" components (i.e.  $\mathcal{T}_{\mathsf D} = \mathcal{T}^{\rm nonid}_{\mathsf{D}} \oplus \mathcal{T}^{\rm id}_{\mathsf{D}}$), such a decomposition can be obtained from a principal connection, in a covariant way. A natural choice of connection is provided by the information geometry, as discussed below. This approach to system identification often appears in the classical setting, and the advantage is that one gains insight in the geometric structure of the parameter manifold, beyond the direct computation of the Fisher information. For instance, the connection can be useful  for developing recursive estimation algorithms based on geodesics of the manifold \cite{hanzon}. For the standard theory of connections on principal bundles, see e.g. \cite{isham}.

In Section \ref{sec.info.geom} we derive the information geometric structure of the statistical estimation problem at hand. Before discussing the statistical aspects, we describe the basic elements of a theory of ``output fluctuations" which is essential for information geometry, but has an interest in its own and deserves to be further investigated. For each  $(k+1)$-tuple of system operators ${\bf X}=(X^0, X^1, \dots, X^k)\in M(\mathbb{C}^d)^{k+1} $ we define the associated 
\emph{fluctuation operator} $\mathbb F_t({\bf X})$ given by the quantum stochastic integral 
$$
\mathbb F_t({\bf X})= \frac {1}{\sqrt t}\int_0^t 
\left(i\sum_{i=1}^{k} j_s(X^i) dA^*_{i}(s)+j_s\circ \mathcal C_{\mathsf D}(X^0) ds\right),\qquad\mathcal C_{\mathsf D}(X):= X-{\rm tr}[\rho^{ss}_{\mathsf D}X]\id,
$$
where $j_s(X):= U_\mathsf{D}^*(s)X U_\mathsf{D}(s)$ is the time-evolved operator $X$. The covariance of $\mathbb F_t({\bf X})$ converges in the limit of large times,  and defines a positive (but degenerate) inner product  on $ M(\mathbb{C}^d)^{k+1}$ (cf. Proposition \ref{continuous_mc} for the explicit formula)
$$
({\bf X}, {\bf Y})_\mathsf{D} = \lim_{t\to \infty}  \langle \mathbb F_t({\bf X})^*\mathbb F_t({\bf Y})\rangle.
$$ 
Furthermore, in Propositions \ref{continuous_mc} and \ref{projection} we construct a linear map  
$R_\mathsf{D}:M(\mathbb{C}^d)^{k+1} \to M(\mathbb{C}^d)^{k+1}$ such that $R_\mathsf{D}$ is a projection onto the subspace of operators of the form $(0, Y^1,\dots, Y^k)\in M(\mathbb{C}^d)^{k+1}$, and the kernel of $R_\mathsf{D}$ is the subspace of degenerate vectors of the inner product. With this definition, 
the inner product take the following simple form 
$$
({\bf X}, {\bf Y})_\mathsf{D}  = \sum_{i=1}^k {\rm tr}\left[ \rho_{ss} R_\mathsf{D} ({\bf X})^{i*} R_\mathsf{D} ({\bf Y})^{i}   \right].
$$
We denote by $\dot{\mathsf{D}} = (\dot{H},  \dot{L}^1,\ldots , \dot{L}^k)$ an element of the tangent space 
$\mathcal{T}_\mathsf{D}$. The \emph{real} linear map ${\bf X}_{\mathsf D}$ defined below plays an important role in connecting fluctuation operators with the information geometry
\begin{eqnarray*}
{\bf X}_{\mathsf D}: \mathcal{T}_{\mathsf{D}} &\to& M(\mathbb{C}^d)^{k+1} \\
\dot{\mathsf{D}} 
&\mapsto&  (\mathcal{E}_{\mathsf D}(\dot{\mathsf{D}}),  \dot{L}^1,\ldots , \dot{L}^k)
\end{eqnarray*} 
where $\mathcal{E}_{\mathsf{D}}$ is the map
\begin{eqnarray*}
\mathcal{E}_{\mathsf{D}}: \mathcal{T}_\mathsf{D}&\to& M(\mathbb{C}^d)\\
\mathcal{E}_{\mathsf{D}}  :\dm& \mapsto &
\dot H+\mathrm{Im} \sum_{i=1}^k \dot L^{i*}L^i .
\end{eqnarray*}
Here the second term in $\mathcal{E}_{\mathsf{D}}$ is due to quantum Ito calculus, hence in some sense represents the effects of the stochastic output on the information geometry. Using the map ${\bf X}_\mathsf{D}$ we define a \emph{real} inner product on the tangent space $\mathcal{T}_{\mathsf{D}}$
$$
(\dot{\mathsf{D}}, \dot{\mathsf{D}}^\prime ) \longmapsto 
{\rm Re} ( {\bf X}_{\mathsf{D}} (\dot{\mathsf{D}}),{\bf X}_{\mathsf{D}} ( \dot{\mathsf{D}}^\prime) )_\mathsf{D}. 
$$
Moreover, since ${\bf X}_\mathsf{D}$ is injective, we can use it to define a projection 
$
P_{\mathsf D}= {\bf X}_{\mathsf D}^{-1}\circ R_{\mathsf D}\circ {\bf X}_{\mathsf D}
$
acting on $\mathcal{T}_{\mathsf{D}}$. Its kernel is the vertical space $\mathcal{T}_{\mathsf{D}}^{\rm nonid}$ whose vectors correspond to infinitesimal ``gauge transformations", and are the degenerate vectors of the inner product. The range of $P_{\mathsf D}$ consists of tangent vectors satisfying the condition $\mathcal E_{\mathsf{D}}(\dm)=0$.
In particular, since $P_{\mathsf D}$  is a projection, the tangent space can be decomposed into ``identifiable'' and ``non-identifiable'' directions (see right panel of Figure \ref{fig.identifiability.classes}) 
$$
\mathcal{T}_\mathsf{D} ={\rm ran} P_\mathsf{D}\oplus {\rm ker} P_\mathsf{D}= 
\mathcal{T}_\mathsf{D}^{\rm id}\oplus \mathcal{T}_\mathsf{D}^{\rm nonid}
$$
whose statistical interpretation is discussed below. This split has also an interesting differential geometric interpretation: the above tangent space decomposition, and inner product are covariant with respect to the action of the group $G$ and define a connection on the resulting principal $G$-bundle, with the associated Lie algebra valued one-form
$$
\omega_{\mathsf D}: \mathcal T_{\mathsf D}\to \mathfrak g, \quad \omega_{\mathsf D}(\dot{\mathsf D})=(\mathbb W_{\mathsf D}^{-1}\circ \mathcal E_{\mathsf D}^0(\dot{\mathsf D}), {\rm tr}[\rho_{ss}\mathcal E_{\mathsf D}(\dot{\mathsf D})])
$$
explicitly depending on the map $\mathcal E_{\mathsf D}(\dot{\mathsf D})$ containing the essential quantum Ito correction. Moreover, the strictly positive inner product on $\mathcal{T}_\mathsf{D}^{\rm id}$ induces a strictly positive  inner product on the tangent space $\mathcal{T}_{[\mathsf{D}]}$ to the point $[\mathsf{D}]$ in the base space $\mathcal{P}= \mathcal{D}^{erg}/G$ of identifiable parameters. As we will see below, this \emph{Riemannian metric} is closely connected to the quantum Fisher information rate of the output state, so we will refer to it as the \emph{information geometry} of the open quantum system, in analogy to the classical case \cite{Amari&Nagaoka}.

Let us consider now the problem of estimating the dynamical parameter $\mathsf{D}$. Although the key constructions could be introduced in a ``coordinate free" way, in order to emphasise the statistical aspects we choose to work with a given (but arbitrary) parametrisation 
$
\theta\mapsto \mathsf{D}_\theta 
$
of $\mathcal{D}^{erg}$, where $\theta$ is an unknown parameter belonging to an open subset of $\mathbb{R}^{m}$, with $m:={\rm dim}(\mathcal{D}^{erg})$. 
At a given point $\mathsf{D}=\mathsf{D}_\theta\in \mathcal{D}^{erg}$, we define the tangent vectors 
$ \dot{\mathsf{D}}_{a} := \partial \mathsf{D}/\partial \theta_a= (\dot{H}_{a}, \dot{L}^1_{a}, \dots, \dot{L}^k_{a}) $ describing infinitesimal changes of the coordinate $\theta_a$, for $a=1,\dots, m$; these vectors form a basis of the tangent space $\mathcal{T}_\mathsf{D}$.

We consider now the $m \times m$ \emph{quantum Fisher information} (QFI) matrix $F^\theta(t)$ associated to the system-output state $|\Psi^{\rm s+o}_{\mathsf{D}_\theta}(t)\rangle$. The QFI is proportional to the real part of the covariance matrix of (centred) ``generators" $G^{0}_{\theta,a}(t) $ of infinitesimal changes with respect to parameter component $\theta_a$ \cite{Braunstein&Caves}.  We show that the generator $G^{0}_{\theta,a}(t) $ (normalised by $t^{-1/2}$) can be expressed as a \emph{fluctuation operator} 
$\mathbb F_t({\bf X}_{\mathsf D} (\dot{\mathsf{D}}_{a}))$, using the map ${\bf X}_{\mathsf{D}}$ defined above. 
As consequence, the QFI grows linearly in time, and the  \emph{QFI rate} per time unit  
$f^\theta= \lim_{t\to \infty} F^\theta(t)/t$ can be expressed in terms of the Markov covariance as
\begin{eqnarray*}
f^\theta_{a,b}&=& 4 {\rm Re} ( {\bf X}_{\mathsf{D}} (\dot{\mathsf{D}}_a),{\bf X}_{\mathsf{D}}( \dot{\mathsf{D}}_b) )_{\mathsf{D}}
= 4 {\rm Re} ( R_\mathsf{D}{\bf X}_{\mathsf{D}} (\dot{\mathsf{D}}_a), R_\mathsf{D}{\bf X}_{\mathsf{D}}( \dot{\mathsf{D}}_b) )_{\mathsf{D}} \\
&=& 
4\sum_{i=1}^k {\rm Re} \, \mathrm{tr}
\left[\rho_{ss}^{\mathsf{D}}  
\left(\dot{L}^i_{a} - i[L^i, \mathbb{W}_\mathsf{D}^{-1}\circ \mathcal{E}^0_{\mathsf D}(\dot{\mathsf{D}}_{a})  ]\right)^*\cdot
\left(\dot{L}^i_{b} - i[L^i , \mathbb{W}_\mathsf{D}^{-1}\circ\mathcal{E}^0_{\mathsf D}(\dot{\mathsf{D}}_{b})  ]\right) \right].
\end{eqnarray*}
where  $\mathbb{W}_\mathsf{D}$ is the Lindblad operator at $\mathsf{D}$. In particular,  the Fisher information rate associated to directions in vertical bundle $\mathcal{T}_{\mathsf{D}}^{\rm nonid}$ (gauge transformations) is equal to zero as expected from the invariance of the output state. This follows from the fact that ${\bf X}_\mathsf{D}$ maps $\mathcal{T}_{\mathsf{D}}^{\rm nonid}$ into ${\rm ker}R_{D}$. 


Above we saw that the \emph{real part} of the Markov covariance $(\cdot , \cdot)_\mathsf{D}$ defines a positive definite inner product on the \emph{real} space $\mathcal{T}^{\rm id}_\mathsf{D}$. In fact, $\mathcal{T}^{\rm id}_\mathsf{D}$ can be made into a complex space by introducing the \emph{complex structure} 
\begin{eqnarray}
\mathcal{J}_{\mathsf D}:\mathcal{T}^{id}_{\mathsf D} & \to & \mathcal{T}^{id}_{\mathsf D}\nonumber\\
\mathcal{J}_{\mathsf D}:  (\dot H, \dot L^1,\ldots,\dot L^k) &\mapsto & \left(\sum_{i=1}^k {\rm Re} \dot L^{i*}L^i\, ,\, i\dot L^1,\ldots,i\dot L^k\right).
\end{eqnarray} 
With this definition the map ${\bf X}_\mathsf{D}$ becomes an isomorphism of complex spaces and 
$(\cdot , \cdot)_\mathsf{D}$ defines a complex inner product on $(\mathcal{T}^{\rm id}_\mathsf{D},\mathcal{J}_{\mathsf D})$. Using the imaginary part $\sigma_D$ of the inner product, we define the canonical commutation relations (CCR) algebra $CCR(\mathcal{T}_{\mathsf{D}}^{\rm id}, \sigma^{\mathsf{D}})$ generated by Weyl operators with commutation relations 
$$
W(\dm)W(\dm')=e^{i\sigma^{\mathsf D}(\dm,\dm')}W(\dm+\dm'), \quad W(-\dm)=W(\dm)^*, \quad
\dot{\mathsf D}, \dot{\mathsf D}' \in \mathcal{T}^{\rm id}_{\mathsf D}.
$$
Following a standard construction we define the Fock representation and the coherent states 
$|\dot{\mathsf{D}} \rangle :=W(\dm)|0\rangle$, where $|0\rangle$ is the vacuum state 
$\langle 0|W(\dm)|0\rangle = \exp(- (\dm, \dm)_\mathsf{D}/2 )$
This model will be interpreted below as limit of the output state model for large times.

Section \ref{sec.example} details the above constructions in the case of special one dimensional models, 
and for a general multidimensional model for a two dimensional system.

\begin{figure}[h]
\includegraphics[width=12cm]{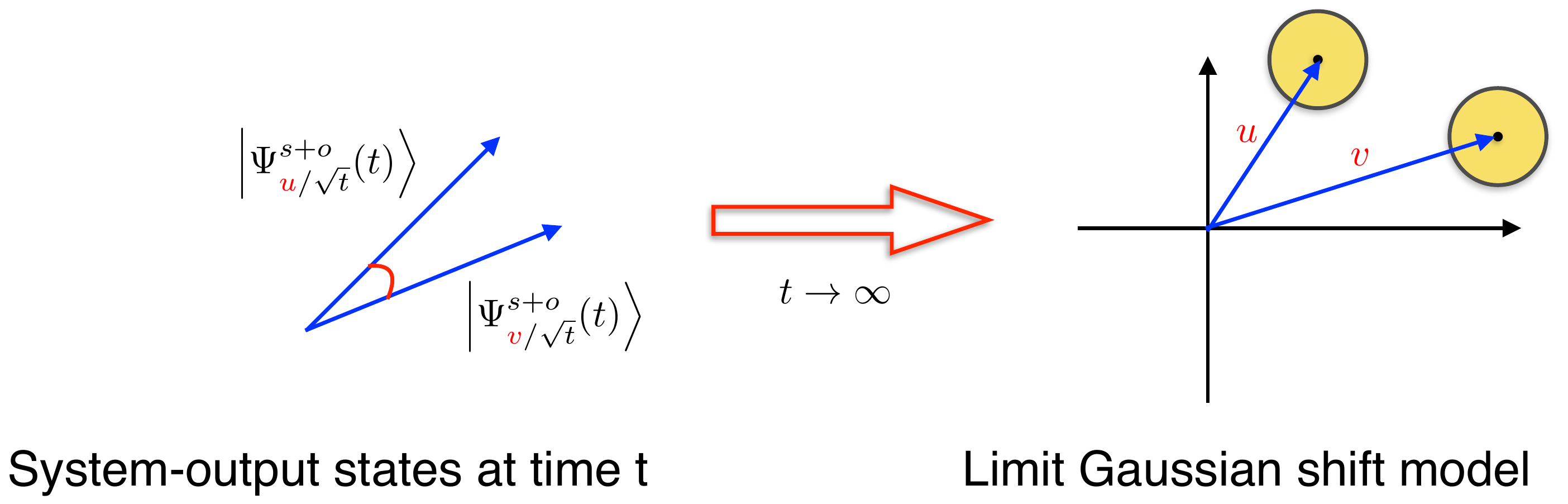}
\caption{Local asymptotic normality as weak convergence to the Gaussian limit. 
The inner products of system-output states with local parameter $u,v$ converge (uniformly) to the inner products of the corresponding coherent state $|u\rangle$ and $|v\rangle$, in the limit of large times.}\label{fig.LAN}
\end{figure}

In section \ref{sec.lan} we study the asymptotic statistical structure of the output state. The main result is the \emph{local asymptotic normality} (LAN) Theorem \ref{th.lan} which shows that both the system-output state, and the stationary output state can be approximated by coherent states of the CCR algebra 
$CCR(\mathcal{T}_{\mathsf{D}}^{\rm id}, \sigma_{\mathsf{D}})$. Below we give a brief description of the result and its interpretation.

Let us consider a parametrisation $u\mapsto [\mathsf{D}]_u$ of (an open subset of) the space of identifiable parameters $\mathcal{P} = \mathcal{D}^{erg}/G$, such that the origin $u=0$ corresponds to a given parameter of interest $[\mathsf{D}_0] = [\mathsf{D}]_{u=0}$. We define the time-indexed family of ``local" statistical models 
$$
\tilde{\mathcal{Q}}_t:= \left\{ \rho^{\rm out}_{u/\sqrt{t}}(t) :  u\in \mathcal{O} \subset \mathbb{R}^{{\rm dim}(\mathcal{P})} 
\right\}
$$
which consist of the output state for unknown parameter values $u/\sqrt{t}$ in a shrinking ball of size scaling as the statistical uncertainty. The local model $\tilde{\mathcal{Q}}_t$ captures the asymptotic properties of the quantum output state for parameters in a neighbourhood  of $[\mathsf{D}_0]$, and it can be justified operationally by means of adaptive procedures whereby a ``small part" of the output can be used to localise the parameter, while the ``remaining part" can be used for estimating the local parameter $u$ (cf. \cite{GutaKahn} for a similar argument in the state estimation setup). 

To define the system-output statistical model let us consider  a 
\emph{horisontal section} $s: \mathcal{P} \to \mathcal{D}^{erg}$ of the principal bundle, i.e. the tangent space to 
$s (\mathcal{P})$ at $\mathsf{D}$ is the horisontal space $\mathcal{T}_{\mathsf{D}}^{id}$. Let 
$$
\mathcal{Q}_t:= \left\{ \left|\Psi^{\rm s+o}_{u/\sqrt{t}}(t)\right\rangle :  u  \in\mathcal{O}\subset \mathbb{R}^{{\rm dim}(\mathcal{P})}  \right\},
$$
be the quantum statistical model  where $\left|\Psi^{\rm s+o}_{u/\sqrt{t}}(t)\right\rangle $ is the joint system-output state at time $t$ at the dynamical parameter $\mathsf{D}_u = s([D]_u)$. The reasons for using a horisontal section in defining the model are as follows. While the stationary output state depends only on the identifiable parameters in $\mathcal{P}$, the system-output state is also sensitive to the location of the parameter within an orbit. 
It turns out that the asymptotic properties can be captured most transparently by choosing a horisontal section which sets certain unphysical phase factors to zero and allows us to understand the model directly in terms of the geometric properties of the vector state,  as explained below. 

Finally, we define the Gaussian model
$$
\mathcal{G}:= \left\{  \rho_u := |u\rangle \langle u | :   u\in \mathcal{O}\subset  \mathbb{R}^{{\rm dim}(\mathcal{P})}  \right\} 
$$
where $|u\rangle =W( \sum_a u_a \dot{\mathsf{D}}_a )|0\rangle$ is the coherent state of the CCR algebra 
$CCR(\mathcal{T}^{id}_{\mathsf D_0}, \sigma^{\mathsf D_0})$. By construction, the QFI of this model is equal  
to the QFI rate $f$ of the output state at $[\mathsf{D}_0]$, which also the QFI of the output state with respect to the \emph{local} parameter $u$, rather than the ``true" parameter $u/\sqrt{t}$.

The first version of LAN states that $\mathcal{Q}_t$ converges (weakly) to $\mathcal{G}$ for large times in the sense of convergence of the inner products (uniformly in $u,v\in \mathcal{O}$), as illustrated in Figure \ref{fig.LAN}
$$
\lim_{n\to\infty} \left\langle \Psi^{\rm s+o}_{u/\sqrt{t}}(t)\right|\left.\Psi^{\rm s+o}_{v/\sqrt{t}}(t)\right\rangle =  \langle u| v \rangle. 
$$
Since pure state models are fully characterised by the inner products, the convergence simply means that for large times the geometry of the system-output states is very similar to that of the coherent states. Although intuitive, this notion of convergence is not suitable for mixed states such as that of the output, and does not have a direct operational meaning.  In the second version of LAN we show that the output models $\tilde{\mathcal{Q}}_t$ converge strongly 
to $\mathcal{G}$ in the sense that there exist channels $T_t$ and $S_t$ such that 
\begin{eqnarray*}
&&\lim_{t\to\infty} \sup_{u \in \mathcal{O}} \left\| T_t \left( \rho^{\rm out}_{u/\sqrt{t}} (t) \right) - \rho_u \right\|_1 =0\\
&&\lim_{t\to\infty} \sup_{u \in \mathcal{O}} \left\| S_t \left(\rho_u\right) -  \rho^{\rm out}_{u/\sqrt{t}} (t) \right\|_1 =0.
\end{eqnarray*}

A concrete consequence of this ``convergence to Gaussianity" is that the QFI computed above is asymptotically ``achievable" in the sense that the estimation of dynamical parameters reduces to that of estimating a Gaussian displacement family with QFI equal to $f_{a,b}$. Similarly to LAN for ensembles of identical states \cite{GutaKahn}, the result implies that the optimal measurement is a linear one (i.e. of homodyne and heterodyne type) and the errors are normally distributed. However, since this paper concentrates on the structure of the quantum states, the measurement and estimation procedures are not discussed here.

%

\section{Preliminaries on Quantum Markov processes}\label{sec.Markov.dynamics}

We begin by introducing notations and necessary background about the input-output formalism of continuous-time quantum Markov processes \cite{GardinerZoller}. The formalism describes the \emph{joint unitary evolution} of an open quantum system interacting with a Bosonic environment in the Markov regime,  cf. Figure \ref{fig.markov}. From this, one can derive the reduced (master) dynamics of the system, as well as the stochastic Schr\"{o}dinger equations for quantum trajectories, describing the stochastic evolution of the system conditional on observations produced by a continuous-time measurement on the environment.  However, in this paper we will be mainly interested in the \emph{output quantum state}, i.e. the state of the environment after the interaction with the system.

Throughout the paper we assume the system to be finite-dimensional, with Hilbert space $\mathcal H=\mathbb C^d$, and associated  algebra of observables $\mathcal A=M(\mathbb C^d)$. 
As we will detail below, the dynamics is specified by the \emph{system Hamiltonian} $H$, together with the 
\emph{quantum jump operators} $ L^{1} ,\dots,L^{k}$. We denote these collectively by
$$
\mathsf{D} =(H,L^1,\ldots,L^k) = (H, {\bf L} )
\in \manif:=M_{sa}(\mathbb C^d)\times M (\mathbb C^d)^k,
$$
and refer to each $\mathsf{D}\in \manif$ as a \emph{dynamical parameter}, and to $\manif$ as manifold (space) of dynamical parameters.

%
\subsection{Environment as quantum noise.} In the Markov approximation, the interaction can be described as the unitary scattering of incoming vacuum Bosonic fields, caused by the continuous interaction with the system.
The environment is modelled by $k$ Bosonic channels  whose Hilbert space is the Fock space 
$$
\mathcal{F}:= \mathcal{F}(\mathfrak{h}_k) =  \mathbb{C}|\Omega\rangle \oplus \bigoplus_{l=1}^\infty  \mathfrak{h}_k^{\otimes_s l}
$$ 
where $\mathfrak{h}_k:= L^2(\mathbb{R}_+)\otimes \mathbb{C}^k$ is the one particle space of the $k$ channels, and 
$|\Omega\rangle$ is the vacuum vector. Similarly, we denote by $\mathcal{F}_{(a,b)}$ the Fock space over 
$L^2((a,b)) \otimes \mathbb{C}^k$. For each time $t$, the symmetric Fock space decomposes as tensor product 
$\mathcal{F}= \mathcal{F}_{(0,t)}\otimes \mathcal{F}_{(t,\infty)}$ between the space of excitations up to time $t$ (the past) and after time $t$ (the future). The fundamental environment degrees of freedom are the annihilation and creation operators of the $i$th channel 
$A_i(f):= A(|f\rangle\otimes |i\rangle)$ and respectively $A^*_i(g):= A_i(g)^*$, which are defined in a standard way \cite{Parthasarathy} for all $|f\rangle ,|g\rangle\in L^2(\mathbb{R}_+)$, and satisfy the commutation relations 
$$
[A^*_i(g) , A_j(f)] = {\rm Im} \langle f|g\rangle \delta_{i,j}\id.
$$
In particular, we will deal with the annihilation and creation \emph{processes} $A_i(t):= A_i(\chi_{[0,t]})$ and 
$A^*_i(t)$,  where $t\in \mathbb{R}_+$ represents time \cite{GardinerZoller,Parthasarathy}. 
These processes are the quantum analogue of the ``classical" Wiener process 
and can be used to define \emph{quantum stochastic integrals} of the form
$$
I(t) = \int_0^t  \sum_{i=1}^k \left[M^i(s) dA_i(s) + N^i(s) dA^*_i(s) \right]+ P(s) ds
$$
where $M^i(s), N^i(s), P(s)$ are \emph{time-adapted} operator valued integrands, i.e. they are of the form $X(s)\otimes \id_{[s,\infty)}$ with respect to the decomposition $\mathcal{F}= \mathcal{F}_{(0,s)}\otimes \mathcal{F}_{(s,\infty)}$. Quantum stochastic integrals can be formally multiplied, and the product $I_1(t) I_2(t)$ is a stochastic integral whose  increment is given by
\begin{equation}\label{equation.product.integrals}
d (I_1(t) I_2(t)) =  d  I_1(t) \cdot I_2(t) +      I_1(t) \cdot  d  I_2(t) +  d  I_1(t) \cdot d  I_2(t)  
\end{equation}
where the third terms is the \emph{Ito correction} which can computed by using the \emph{quantum Ito rule}  
\begin{equation}\label{itoformula}
dA_i(t) dA_j^*(t) =\delta_{i,j} dt
\end{equation}
while all other products are zero.

\subsection{Interaction as input-output scattering.}
We now introduce the coupling between system and the Bosonic environment, and the corresponding unitary evolution. Each dynamical parameter $\mathsf{D}= (H, \mathbf{L})$ determines a unique continuous family $U(t)$ of unitary operators (cocycles) describing evolution of the system and environment in the interaction picture with respect to the free evolution of the fields, the latter being given by the second quantisation of the right shift on $L^2(\mathbb{R})$. The unitaries are defined as the solution of the quantum stochastic differential equation (QSDE)
\begin{align}\label{QSDE}
dU(t) = \left(\sum_{i=1}^{k}(L^i \otimes dA_i^*(t)-L^{i*}\otimes dA_i(t))-iH_{\rm eff} \otimes \id_\mathcal{F} \, dt\right)U(t), 
\end{align}
with initial condition $U(0)=\id$. Here, $H_{\rm eff}$ is the effective Hamiltonian $H_{\rm eff}:=H-\frac{i}{2}\sum_{i=1}^{k} L^{i*}L^i$ which generates a semigroup $S(t)=e^{-it H_{\rm eff}}$ of contractions on the system's space, and describes the evolution of the system between consecutive quantum jumps. The imaginary part $\frac{i}{2}\sum_{i=1}^{k} L^{i*}L^i$ is the Ito correction which insures that $U(t)$ is unitary. For simplicity of notation, from now on we will omit the tensor product and simply write $L^i \otimes dA_i^*(t)$ as 
$L^i dA_i^*(t)$, an similarly for other integrands.


If the system is initialised in state $|\varphi\rangle$ and the input fields are in the vacuum state $|\Omega\rangle$, then the state of the system together with the output after time $t$ is given by
\begin{equation}\label{eq.output}
|\Psi^{\rm s+o}(t)\rangle = U(t)|\varphi\otimes \Omega\rangle = V(t)|\varphi\rangle,
\end{equation}
where $V(t):\mathcal H\to \mathcal H\otimes \mathcal F$ is a family of isometries defined by the second equality. 
The \emph{output state} is the state of the scattered field modes after the interaction with the system, and is obtained by tracing out the system
$$
\rho^{\rm out}(t) ={\rm tr}_{\mathcal{H}}[|\Psi^{\rm s+o}(t)\rangle\langle \Psi^{\rm s+o}(t)|].
$$
%
Let us denote by $j_t(X):=U(t)^*(X\otimes \id_{\mathcal F})U(t)$ the Heisenberg evolved system operator $X$. 
Using equation \eqref{QSDE}, we find that the operators satisfy the \emph{quantum Langevin equation}
\begin{equation}\label{langevin}
dj_t(X) =\sum_{i}\left(j_t([X,L^i])dA^*_{i}(t)+j_t([L^{i*},X])dA_{i}(t)\right)+j_t(\mathbb W(X)) dt.
\end{equation}
where 
$$
\mathbb W(\cdot)=-i(\cdot)H_{\rm eff}+iH_{\rm eff}^*(\cdot)+\sum_{i=1}^{k} L^{i*}(\cdot)L^i
$$
is called the \emph{Lindblad generator}. Its significance can understood by considering the \emph{reduced Heisenberg evolution} of the system $T_t: \mathcal{A}\to \mathcal{A}$ defined by taking the expectation over the environment 
\begin{equation}\label{channel}
T_t(X):= \langle \Omega | j_t(X) |\Omega\rangle=V(t)^*(X\otimes \id_{\mathcal F})V(t).
\end{equation}
From \eqref{langevin} we find
$
d T_t(X) =d \langle \Omega |  j_t(X) |\Omega\rangle = 
\mathbb{W} (T_t (X))
$
which means that $T_t$ is a trace preserving completely positive semigroup with generator $\mathbb{W}$. The generator is said to be \emph{ergodic}, if it has a unique stationary state $\rho_{ss}$ (i.e. $\mathbb W_*(\rho_{ss})=0$) which has full rank. In this case \cite{Frigerio}
\begin{align}\label{limitchannel}
\lim_{t\rightarrow\infty} T_t&=\lim_{t\rightarrow\infty} \frac{1}{t} \int_0^t T_s ds = \mathrm{tr}[\rho_{ss}(\cdot)]\id.
\end{align}
Since $\rho_{ss}$ is a stationary state, the range of $\mathbb W$ is included in
$$\mathcal B_0=\{ X\mid \mathrm{tr}[\rho_{ss} X]=0\}.$$
By ergodicity, $\id$ is the only fixed point of $e^{t\mathbb W}$, and hence $\ker ( \mathbb W )$ is spanned by $\id\notin \mathcal B_0$. This implies that the range has dimension $d^2-1=\dim \mathcal B_0$, i.e. $\mathbb W$ is surjective onto $\mathcal B_0$, and the restriction of $\mathbb W$ onto $\mathcal B_0$ is injective. Hence $\mathbb W$ is invertible on $\mathcal B_0$, and we let $\mathbb W^{-1}:\mathcal B_0\to\mathcal B_0$ denote the inverse. Furthermore, the following limit exists:
\begin{equation}\label{inverselimit}
-\mathbb W^{-1}=\lim_{t\rightarrow\infty} \int_0^t T_s ds.
\end{equation}
%
%
%

\section{Identifiability of continuous quantum Markov processes}\label{sec.identifiability}
This section deals with the problem of characterising the equivalence classes of Markov dynamics with identical stationary output states. We restrict ourselves to ergodic Markov processes, although similar results are expected to hold more generally.  Similar results have been obtained in \cite{GutaKiukas} for \emph{discrete time} quantum Markov processes. 

\newcommand{\manifprim}{\mathcal D^{\rm erg}}

Let $\manifprim$ denote the open subset of $\mathcal D$ consisting of dynamical parameters for which the associated Markov process is ergodic; this will be the relevant paramter set for subsequent considerations.

\begin{definition} Two dynamical parameters $\mathsf{D},\mathsf{D}'\in \manifprim$ are
\emph{output-equivalent} if the stationary output states \eqref{eq.output} of the associated continuous-time Markov processes are identical. We denote the set of associated equivalence classes by
$
\mathcal{P}:= \{[\mathsf{D}] ~:~ \mathsf{D}\in \mathcal{D}^{\rm erg}\}.
$
\end{definition}
Of course, the same equivalence can be formulated for arbitrary parameters in $\mathcal D$. However,  it turns out that 
when restricted to  $\manifprim$, the equivalence classes have a simple characterisation in terms of the following transformations:
\begin{itemize}
\item[(PM)] Phase conjugation on the Hamiltonian: 
$$
(H,L^1,\ldots,L^k)\mapsto (H+r \id ,L^1,\ldots,L^k), \quad (r\in \mathbb R).
$$
\item[(UC)] Conjugation by system unitary $W$:
$$
(H,L^1,\ldots,L^k)\mapsto (W^*HW,W^*L^1W,\ldots,W^*L^kW).
$$
\end{itemize}

Indeed, it is easy to verify that (PM) and (UC) do not change the output of the associated continuous Markov process. The following Theorem shows that the converse is also true.  The details of the proof can be found in Appendix \ref{proof.identifiability.ct}.

\begin{theorem}\label{equiv_thm}
Let $\mathsf{D},\mathsf{D}'\in \manifprim$. Then $\mathsf{D}$ and $\mathsf{D}'$ are output-equivalent if and only if they can be obtained from each other via the transformations (UC) and (PM).
\end{theorem}
The interpretation of the result is that parameters along the equivalence classes described by the transformations (UC) and (PM) are not identifiable, while the identifiable parameters are ``transversal" to these classes, as illustrated in Figure \ref{fig.identifiability.classes}. It is now convenient to formulate the equivalence classes in terms of an action of the appropriate Lie group $G:= PU(d)\times \mathbb{R}$. For that we regard $\mathcal D^{\rm erg}$ as an open submanifold of $\mathcal D$, and use transformations (PM) and (UC) to set up the action:
\begin{align}\label{gaction}
&G\times \mathcal{D}^{\rm erg} \to \mathcal{D}^{\rm erg}\nonumber\\
&(g,\mathsf D) \mapsto g\mathsf D :=( W^*HW , W^* L^1 W, \ldots ,  W^* L^k W ) + a(\id, 0,\ldots ,0), \\
& g = (W,a)\in PU(d) \times \mathbb{R}, \quad \mathsf D = ( H , L^1, \ldots ,  L^k  )\in \mathcal D^{\rm erg}.\nonumber
\end{align}
Here $PU(d)=U(d)/U(1)$ is the projective unitary group, equipped with its unique Lie group structure, and the above action is defined as the natural lift of the corresponding one for $U(d)$. The reason to use $PU(d)$ instead of $U(d)$ will become clear from the proof of the Lemma below.

The above theorem implies that the equivalence class $ [\mathsf D]\in \mathcal{P}$ is the orbit of $D\in \mathcal{D}^{\rm erg}$ under the action of $G$, such that $\mathcal{P}$ can be identified with the quotient $\mathcal D^{\rm erg}/G$. The following lemma which relies on the ergodicity assumption, is essential for understanding the structure of the quotient, as we will see below.
In order to avoid confusion with the output equivalence, we identify $W\in PU(d)$ with a representative unitary operator without explicit indication. 


\begin{lemma}
The Lie group action $G\times \mathcal D^{\rm erg}\to \mathcal D^{\rm erg}$ is smooth, proper, and free.
\end{lemma}
\begin{proof} The action 
defined via \eqref{gaction} is clearly smooth with respect to $W,$ and $a$, hence its lift to the quotient Lie group $PU(d)\times \mathbb R$ is smooth as well. Since the group $PU(d)$ is compact, and the rest is just a translation, it follows from elementary arguments that the smooth map $G\times \mathcal D^{\rm erg}\to \mathcal D^{\rm erg}\times \mathcal D^{\rm erg}$ given by $(g,\mathsf D)\mapsto (g\mathsf D,\mathsf D)$ is proper, i.e. preimage of every compact set is compact. This means that the action is proper. In order to show that the action is free, we need to use ergodicity as follows: suppose that $g\mathsf D =g'\mathsf D$ for some $\mathsf D$, and $g=(W,a)$, $g'=(W',a')$; then a direct computation similar to the one in the proof of Lemma \ref{firstlemma} in Appendix shows that $\mathbb W_{\mathsf D}(W^*W')=i(a-a')W^*W'$. Since ergodicity requires
$$
\lim_{\tau\rightarrow \infty} e^{i\tau (a-a')} W^*W' = {\rm tr}[W^*W'\rho_{ss}]\mathbb I,
$$
we must have $a=a'$ and $W^*W'$ a multiple of the identity. But this exactly means that $W$ equals $W'$ as an element of the \emph{projective} unitary group; hence $g=g'$. This proves that the action is free.
\end{proof}
The fact that the group action preserves the equivalence classes, that is $g\mathsf D \in [\mathsf D]$ for all $\mathsf D\in \mathcal D^{\rm erg}$, can now be formulated in differential geometric terms. Indeed, using the standard theory of Lie group actions on manifolds, we conclude from the above Lemma that the space of output equivalence classes 
$$
\mathcal{P}= \{[\mathsf{D}] ~:~ \mathsf{D}\in \mathcal{D}^{\rm erg}\} =  \mathcal{D}^{\rm erg}/G
$$
admits a unique smooth structure such that the quotient map $\pi: \mathcal{D}^{\rm erg}\to \mathcal{P}$ 
$$
\pi(\mathsf D) = [\mathsf D], \quad \text{ for all } \mathsf D\in \mathcal D^{\rm erg},
$$ 
is a submersion, and $\mathcal{D}^{\rm erg}$ is a principal $G$-bundle over $\mathcal{P}$ \cite{isham}. Here the equivalence classes $[\mathsf D]$ are considered as \emph{fibres} of the \emph{fiber bundle} over the \emph{base manifold} $\mathcal P$, that is, the map $\pi$ has the local triviality property: each $[D]\in  \mathcal{P}$ has an open neighbourhood $U$ such that there exists a diffeomorphism
$$
\phi: \pi^{-1} (U) \to U\times G,
$$
which is $G$-equivariant, i.e. $\phi(g \mathsf{D}) = g \phi(\mathsf{D})$ where $G$ acts on $U\times G$ as 
$g ([\mathsf D], g^\prime) := ([\mathsf D] , g^\prime g^{-1})$. The term principal $G$-bundle refers to the fact that the group action preserves the fibres.

We can now use this differential geometric framework to describe local changes of identifiable parameters, via the tangent bundle $\mathcal T$ of the manifold $\manifprim$. In particular, the \emph{non-identifiable} parameter changes along the equivalence classes correspond to the \emph{vertical bundle} over $\manifprim$ with the fibres
\newcommand{\nonident}{\mathcal T^{\rm nonid}}
$$
\mathcal T_{\mathsf D}^{\rm nonid} := \ker \left. \pi_* \right|_{\mathsf D}\subset \mathcal T_{\mathsf D}, \quad \mathsf D\in \manifprim,
$$
where $\left. \pi_* \right|_{\mathsf D}$ is the push-forward tangent map of the canonical projection $\pi$ at point $\mathsf D$, and 
$\mathcal T_{\mathsf D}$ is the full tangent space at that point. 

The group action is reflected in two ways at the level of tangent spaces. On the one hand, given any fixed $g\in G$, the push-forward $g_*$ of the map $\mathsf D\mapsto g\mathsf D$ maps the fibres into each other as
$$
g_*\nonident_{\mathsf D}=\nonident_{g\mathsf D}.
$$
This push-forward is simply obtained by differentiating the parameters in the standard chart
$$
g_*(\dot H, \dot L^1,\ldots, \dot L^k) = (W^*\dot H W, W^*\dot L^1 W,\ldots, W^*\dot L^kW), \quad g=(W,a).
$$
On the other hand, for any fixed $\mathsf D$, the push-forward of $g\mapsto g\mathsf D$ defines a Lie algebra isomorphism
\begin{equation}\label{eq.Dstar}
\mathsf D_*:\mathfrak g\to \nonident_{\mathsf D},
\end{equation}
so that different fibres all have the same dimension, which is that of the Lie algebra $\mathfrak g$. We can now explicitly compute this action. First of all, the Lie algebra of $G$ can be conveniently written as
\begin{equation}\label{liealg}
\mathfrak g =\{ (-iK,r)\mid K\in M_{sa}(\mathbb C^d)/\mathbb R\id,\, r\in \mathbb R\} = \{ (-iK,r)\mid K\in M_{sa}(\mathbb C^d), \, {\rm tr}[\rho^{ss}K]=0,\, r\in \mathbb R\},
\end{equation}
where the choice of the sign as well as the last identification is for later convenience. In particular, the subspace of non-identifiable directions is in one-to-one correspondence with this linear space. From this we already find the number of non-identifiable directions:
$$
\dim \nonident_{\mathsf D} = d^2-1+1 = d^2.
$$
We stress that this result is crucially based on ergodicity, which ensures that the action is free; this is required for the push-forward $\mathsf D_*$ to be an isomorphism. Now $\mathsf D_*$ acts on an element $X=(-iK,r)\in \mathfrak g$ as
\begin{align*}
\mathsf D_*(X) &= \frac{d}{dt}(\exp(t\,(-iK,r))\mathsf D)|_{t=0}\\
&=\frac{d}{dt}\left(( e^{itK}He^{-itK} , e^{itK} L^1e^{-itK}, \ldots ,  e^{itK} L^k e^{-itK} ) + t\,r(\id, 0,\ldots ,0)\right)|_{t=0}\\
&= (i[H,K], i[L^1,K], \ldots, i[L^k,K])+ r(\id, 0,\ldots ,0).
\end{align*}

\newcommand{\ident}{\mathcal T^{\rm id}}
Having now characterised the vertical bundle $\nonident$ of non-identifiable directions, an obvious question arises: is there a natural way to choose complementary subspaces for \emph{identifiable directions} in each fibre? This means choosing subspaces $\ident_{\mathsf D}$ such that
$$
\mathcal T_{\mathsf D} = \nonident_{\mathsf D}\oplus \ident_{\mathsf D}, \quad \mathsf D\in \manifprim.
$$
If the subspaces are chosen smoothly, i.e. so as to define a fibre bundle $\ident$ over $\manifprim$, the result is called a \emph{horizontal bundle} $\ident$, and in case it respects the group action, that is
\begin{equation}\label{horizontaleq}
g_*\ident_{\mathsf D}=\ident_{g\mathsf D},
\end{equation}
it defines an \emph{principal connection} on the manifold $\manifprim$.

There is a natural way of defining a principal connection via its associated \emph{connection one-form}; since this approach turns out to be relevant in our situation, we briefly explain the idea in the general level. As we have shown above, any $\dot{\mathsf D}\in \nonident_{\mathsf D}$ can be generated by the action of the Lie algebra; $\dot{\mathsf D}=\mathsf D_*(X)$ for some $X\in \mathfrak g$. Now suppose that we can associate to \emph{every} tangent vector $\dot{\mathsf D}\in \mathcal T_{\mathsf D}$ an element $\omega_{\mathsf D}(\dot{\mathsf D})\in \mathfrak g$ which somehow describes the "part" of the parameter that results from the non-identifiable group action. Such a map should define a one-form $\omega:\mathcal T\to \mathfrak g$ satisfying the compatibility condition (sometimes called \emph{nondegeneracy})
\begin{equation}\label{comp}
\dot{\mathsf D}=\mathsf D_*(\omega_{\mathsf D}(\dot{\mathsf D})), \quad \dot{\mathsf D}\in \nonident_{\mathsf D},
\end{equation}
and the $G$-covariance condition
\begin{equation}\label{gcovariance}
g^*\omega={\rm Ad}_{g^{-1}}\circ \omega,
\end{equation}
where $g^*$ is the pull-back of the action by $g$ on the cotangent bundle, which simply acts as $g^*\omega(\dot{\mathsf D}) = \omega_{g\mathsf D}(g_*\dot{\mathsf D})=\omega_{g\mathsf D}(W^*\dot{\mathsf D}W)$ for $g=(W,a)$, and the adjoint action is given by ${\rm Ad}_{g^{-1}}(X)=(W^*XW,r+a)$.

Given such a map, we can then define the "back-action" $\mathsf D_*\circ \omega_{\mathsf D}$ on the tangent space; due to the above compatibility condition, the back-action is a special projection of the tangent space onto the subspace $\nonident_{\mathsf D}$. Hence, we can use its complementary projection
$$
P_{\mathsf D}:= {\rm Id}-\mathsf D_*\circ \omega_{\mathsf D}
$$
to define the above horizontal bundle and the associated principal connection via
$$
\ident_{\mathsf D}:={\rm ran}\, P_{\mathsf D}.
$$
Indeed, the condition \eqref{horizontaleq} holds because of \eqref{gcovariance}. The map $\omega$ is called the \emph{connection one-form}, and $P_{\mathsf D}$ is the \emph{horizontal projection}.

Any principal connection gives a possible way of extracting the parameter changes relevant for our system identification problem. In the next section we show that there is actually a natural connection associated with the \emph{information geometric structure} of the problem, given by the Fisher information of the output state. We will obtain it by explicitly constructing the associated connection one-form, which arises neatly from the quantum Ito calculus.

\section{Information geometry for dynamical parameter estimation from the output state}\label{sec.info.geom}
Our goal is to describe quantitatively the precision with which unknown dynamical parameters can be estimated by making measurements on the output state. As noted above, we will restrict our attention to dynamical parameters $\mathsf{D}$ which belong to the open subset $\mathcal{D}^{\rm erg}$ of $\mathcal{D}$ of \emph{ergodic} Markov dynamics. As we will consider this problem in the limit of large times, the relevant dynamical regime is the stationary one; moreover, the statistical properties of the output state can be understood \emph{locally}, by focusing on a shrinking neighbourhood of the parameter manifold $\mathcal{D}^{\rm erg}$ whose size is of the order of the statistical uncertainty \cite{GutaKiukas,CatanaGutaBouten}. This will lead to the concept of local asymptotic normality discussed in section \ref{sec.lan}. In this section however, we focus on the information geometry of the system identification problem, more precisely on the quantum Fisher information matrix of the output state and its asymptotic behaviour, and its relationship with the covariance of certain quanta stochastic integrals called ``fluctuation operators". We will start by introducing the latter in a  general set-up and then show how the former fits in this theory.

Section \ref{subsec.QFI} derives the quantum Fisher information of the system-output state as covariance of certain ``generators"; section \ref{subsec.fluctuations} analyses more general ``fluctuation operators" and looks at their Markov covariance; section \ref{sec.unidentifiable.param} deals with the information geometry structure, and connects the previous constructions, in particular it  provides an explicit expression of the quantum Fisher information; section \ref{sec.ccr} constructs an algebra of canonical commutation relations (multimode continuous variables system) and a family of coherent states which will be relevant later on for the local asymptotic result.

\subsection{Quantum Fisher information of a parametric model}\label{subsec.QFI}
We pass now to a statistical setting where the dynamical parameter $\mathsf D$ is considered to be unknown.
 The changes in $\mathsf D$ are encoded in its (partial) derivatives $\dm=(\dot H,\dot L^1,\ldots, \dot L^k)$, which will be seen as vectors in the tangent space $\mathcal{T}_{\mathsf D}$ to $\mathcal{D}^{\rm erg}$ at the point ${\mathsf D}$. Since the dynamics is ergodic, 
 the system converges to a unique stationary state $\rho_{ss}$ for large times, and we will denote by  $\langle\cdot \rangle_{ss}$ the expectation with respect to the state $\rho_{ss} \otimes |\Omega\rangle\langle \Omega|$.
In this subsection we consider a generic statistical model and analyse the quantum Fisher  information (QFI) of the output state; we will show that the QFI grows linearly with time and the rate can be expressed in terms of the Markov covariance inner product introduced below. Let
$$
\mathbb R^m \ni \theta\mapsto \mathsf{D}_\theta.
$$
be a smooth family of dynamics parametrised by an unknown parameter  $\theta\in \mathbb{R}^m$ which may be thought to encode our prior knowledge about the dynamics. In particular this could be a complete parametrisation of $\mathcal{D}^{erg}$.
The directional derivatives of $\mathsf{D}_\theta$ are defined as
$$
\dot{\mathsf{D}}_{\theta,a} :=\left (\frac{\partial H}{\partial \theta_a} , \frac{\partial L^1}{\partial \theta_a},\ldots,
\frac{\partial L^k}{\partial \theta_a} \right) = (\dot{H}_{\theta,a}, \dot{L}^1_{\theta,a}, \dots, \dot{L}^k_{\theta,a}) \in \mathcal{T}_{\mathsf{D}_\theta}.
$$

Recall that the QFI of an arbitrary multiparameter (smooth) family of pure states $|\psi_\theta\rangle$ with $\theta\in \mathbb{R}^m$, is the $m\times m$ 
positive real matrix with elements \cite{Braunstein&Caves}
$$
F^\theta_{a,b}= 4 {\rm Re}\left(\left\langle \left.\frac{\partial\psi_\theta}{\partial\theta_a} \right| \frac{\partial \psi_\theta}{\partial \theta_b} \right\rangle -
\left\langle \psi_\theta \left|\frac{\partial \psi_\theta}{\partial \theta_b}\right.\right\rangle
\left\langle \left.\frac{\partial \psi_\theta}{\partial \theta_a}\right|\psi_\theta\right\rangle  \right)
,\quad 1\leq a,b\leq m.
$$
We apply this formula to the output state $|\Psi^{\rm s+o}_\theta(t)\rangle:=U_\theta(t)|\varphi\otimes \Omega\rangle$ generated with a $\theta$-dependent dynamical parameter $\mathsf{D}_\theta$, cf. equation \eqref{eq.output}. By differentiating with respect to $\theta_a$ we get
\begin{equation}\label{eq.derivative.psi}
U^*_\theta (t) \frac{\partial }{\partial \theta_a} \left| \Psi^{\rm out}_\theta(t)\right\rangle  = 
U^*_\theta (t)\dot{U}_{\theta,a}(t) |\varphi\otimes \Omega \rangle, \qquad
\dot{U}_{\theta,a} (t):=\frac{\partial U_\theta(t)}{\partial \theta_a}. 
\end{equation}
We will now show that the \emph{generator} 
$-i G_{\theta, a}(t):= U^*_\theta (t)\dot{U}_{\theta,a}(t)$ can be written as a quantum stochastic integral. 
From \eqref{QSDE} we have
\begin{align*}
dU^*_\theta(t)&=U^*_\theta(t) \left(\sum_i (-L^i_{\theta} dA^*_{i}(t)+L^{i*}_\theta dA_{i}(t))-(-iH_\theta+\frac 12 \sum_i 
L^{i*}_\theta L^i_\theta )dt \right),\\
d\dot U_\theta (t) &= \left(\sum_i (\dot L^i_{\theta,a} dA^*_{i}(t)-\dot L^{i*}_{\theta,a} dA_{i}(t))-(i\dot H_\theta+\frac 12 \sum_i (\dot L^{i*}_{\theta,a} L^i_\theta +L^{i*}_\theta \dot L^i_\theta ))dt\right)U_\theta (t)\\
&+ \left(\sum_i (L^i_\theta dA^*_{i}(t)-L^{i *}_\theta dA_{i}(t))-(iH_\theta +\frac 12 \sum_i L^{i*}_\theta L^i_\theta )dt\right)\dot U_\theta (t).
\end{align*}
Therefore, by applying the Ito rule \eqref{itoformula} we get
$$
dU^*_\theta(t) \cdot d\dot U_{\theta,a}(t) = U_\theta(t)^*\sum_i L^{i*}_\theta \left(\dot L^i_{\theta,a} U(t)+L^i_\theta \dot U_\theta(t)\right)\, dt.
$$
and using \eqref{equation.product.integrals} 
we obtain an explicit differential expression for the generator
\begin{eqnarray}
d G_{\theta,a}(t) &=& i d(U^*_\theta(t)\dot U_{\theta,a}(t))\nonumber \\
&=&i \sum_{i} \left(j_t(\dot L^i_{\theta,a}) dA^*_{i}(t)-j_t(\dot L^{i*}_{\theta,a}) dA_{i}(t)\right)+ 
 j_t\left(\dot H_\theta+\mathrm{Im} \sum_i \dot L^{i*}_{\theta,a} L^i_\theta \right)dt 
\nonumber \\ 
&=&i \sum_{i} \left(j_t(\dot L^i_{\theta,a}) dA^*_{i}(t)-j_t(\dot L^{i*}_{\theta,a}) dA_{i}(t)\right) + 
 j_t\left( \mathcal{E}_{\mathsf{D}}({\dot{\mathsf{D}}_{\theta,a}}) \right)dt
\label{eq.generator.diff}
\end{eqnarray}
where the \emph{real linear} map 
$\mathcal{E}_{\mathsf{D}}: \mathcal{T}_{\mathsf{D}} :\to M_{sa}(\mathbb{C}^d)$ is given by
\begin{equation}\label{eq.E}
\mathcal E_{\mathsf{D}}(\dm):=\dot H+\mathrm{Im} \sum_{i=1}^k \dot L^{i*}L^i.
\end{equation}
Later on we will see that this map turns out to play a crucial role in the construction of the horizontal bundle for the identifiable parameters, and in the definition of the CCR algebra in section \ref{sec.ccr}.

The QFI can be written in terms of the covariance matrix of the generators $G_{\theta,b}(t)$ 
$$
F^\theta_{a,b}(t) = 
4 {\rm Re}\left(\langle \varphi \otimes \Omega | G^*_{\theta,a}(t) G_{\theta,b}(t) | \varphi \otimes \Omega\rangle -
 \langle \varphi \otimes \Omega | G^*_{\theta,a}(t) | \varphi \otimes \Omega\rangle
\langle \varphi \otimes \Omega | G_{\theta,a}(t) |\varphi \otimes \Omega\rangle\right).
$$
where the second term stems from the fact that  $G_{\theta,b}(t)$ have non-zero mean. The generators are in fact not uniquely defined: since $dA_i(t)$ annihilates the vacuum state, arbitrary annihilation integrals can be added, while terms proportional to the identity produce only unphysical complex phases which do not change the state. We will therefore define a modified (non-selfadjoint) generator which ``centres" $G_{\theta,b}(t)$ for large times, and lacks annihilations terms so that it is consistent with the definition of ``fluctuation operators" introduced in the next subsection. The modified generator is given by the quantum stochastic integral with differential form
\begin{equation}\label{eq.modified.generator}
dG^0_{\theta,a}(t) = 
i \sum_{i=1}^k j_t(\dot L^i_{\theta,a}) dA^*_{i}(t) + 
  \left( j_t ( \mathcal{E}_{\mathsf{D}}({\dot{\mathsf{D}}_{\theta,a}}))  -  {\rm Tr}\left(\rho_{\mathsf{D}}^{ss} \mathcal E_{\mathsf{D}}(\dm_{\theta,a}) \right)\right)dt.
\end{equation}
By ergodicity, its rescaled mean converges to zero
$$
\lim_{t\to\infty}\frac{1}{t}
 \langle \varphi \otimes \Omega  |  G^0_{\theta,a}(t) | \varphi \otimes \Omega\rangle =0
$$ 
For large times, the QFI matrix elements scale linearly with $t$ and the leading contribution is given by the 
\emph{quantum Fisher information rate}
\begin{equation}\label{eq.qfi.rate}
f^\theta_{a,b} := \lim_{t\to\infty} \frac{F^\theta_{a,b}(t)}{t} =  
\lim_{t\to\infty} \frac{1}{t} 4
{\rm Re} \langle \varphi \otimes \Omega | G^{0*}_{\theta,a}(t) G^0_{\theta,b}(t) | \varphi \otimes \Omega\rangle .
\end{equation}
In the next section we prove the linear scaling and find an explicit expression of the QFI rate.

\subsection{Fluctuation operators and the Markov covariance form}\label{subsec.fluctuations}

Our goal is now to formulate the QFI rate \eqref{eq.qfi.rate} in terms of certain quantum fluctuation operators, and subsequently compute it using quantum stochastic calculus. These fluctuation operators can be formulated in a slightly more general setting, which is naturally \emph{complex linear} instead of real linear, and is also independent on the map $\mathcal E_{\mathsf D}$ special to our setting. The dynamical parameter $\mathsf D$ will remain fixed throughout the section.

Recall that for any $X\in M(\mathbb C^d)$ we let $j_t (X)$ denote the Heisenberg evolved system observable defined by the Langevin equation \eqref{langevin}. For an arbitrary $(k+1)$-tuple ${\bf X} := (X^0, X^1,\dots, X^k) \in M(\mathbb C^d)^{1+k}$ we define the associated \emph{centered fluctuation operator} by the quantum stochastic integral
\begin{equation}\label{eq.fluctuation}
\mathbb F_t({\bf X})= \frac {1}{\sqrt t}\int_0^t 
\left(i\sum_{i=1}^{k} j_s(X^i) dA^*_{i}(s)+j_s\circ \mathcal C_{\mathsf D}(X^0) ds\right),
\end{equation}
where the map
$$\mathcal C_{\mathsf D}(X):= X-{\rm tr}[\rho_{ss}^{\mathsf D}X]\id
$$
``centers" the \emph{stationary mean} of $\mathbb F_t({\bf X})$ to zero:
$$
\langle \mathbb F_t({\bf X}) \rangle_{ss}
= \frac {1}{\sqrt t}\int_0^t \mathrm{tr}[\rho_{ss}T_s(X^0-{\rm tr}[\rho_{ss} X^0]\id)] ds=\frac {1}{\sqrt t}\int_0^t \mathrm{tr}[\rho_{ss}(X^0-{\rm tr}[\rho_{ss} X^0]\id)] dt=0.
$$
The proof of the following crucial result is based on quantum Ito calculus, and can be found in the Appendix.
\begin{proposition}[{\bf Markov covariance for fluctuation operators}]\label{continuous_mc} The following limit exists, is independent of the unit vector $|\varphi \rangle \in \mathcal H$, 
and defines a positive sesquilinear form $(\cdot,\cdot )_{\mathsf{D}}$ on the complex linear space $M(\mathbb C^d)^{1+k}$ via
\begin{align*}
(\mathbf{X},\mathbf Y)_{\mathsf{D}}&:=\lim_{t \rightarrow\infty} \langle \varphi\otimes \Omega|\mathbb F_t(\mathbf X)^*\mathbb F_t(\mathbf Y)|\varphi\otimes \Omega\rangle=\sum_{i=1}^{k}
{\rm tr}\left[\rho_{ss}  R_{\mathsf D}({\bf X})^{i*} \,R_{\mathsf D}({\bf Y})^{i} \right],
\end{align*}
where
$$
R_{\mathsf D}({\bf X})=(\mathcal{C}_\mathsf{D}(X^0), X^1, \dots , X^k) - \mathcal L_{\mathsf D}\circ \mathbb W_{\mathsf D}^{-1}\circ \mathcal C_{\mathsf D}(X^0), \qquad \mathrm{and}\qquad
\mathcal L_{\mathsf D}(X) = \left(\mathbb W_{\mathsf D}(X),i[L^1,X],\ldots,i[L^{k},X]\right)
$$
We call $(\cdot,\cdot)_{\mathsf{D}}$ the Markov covariance inner product.
\end{proposition}
From this proposition it is clear that the map $R_{\mathsf D}$ plays a central role; in particular, since $\rho_{ss}$ has full rank, the kernel of the Markov covariance coincides with $\ker R_{\mathsf D}$. Also the range of $R_{\mathsf D}$ turns out to be relevant. These subspaces can be characterised explicitly as follows.

\begin{proposition}\label{projection} The operator $R_{\mathsf D}$ is a projection, i.e. $R_{\mathsf D}^2=R_{\mathsf D}$, with range and kernel given by
\begin{align*}
\ker R_{\mathsf D}&=\left\{ (\mathbb W_{\mathsf D}(K)+r\id, i[L^1,K], \ldots, i[L^{k}, K])\,\big| \,K\in M(\mathbb C^d), \, r\in \mathbb C\right\},\\
{\rm ran}\, R_{\mathsf D}&= \left\{ \left(0,  Y^1 , \ldots , Y^k \right)   \, \Big| \,  
Y^1,\dots , Y^k \in M(\mathbb{C}^d) \right\}.
\end{align*}
\end{proposition}
\begin{proof} First of all, ${\bf X}\in \ker R_{\mathsf D}$ if and only if $X^i=i[L^i,\mathbb W^{-1}(X^0-{\rm tr}[\rho_{ss}X^0]\id)]$ for all $i=1,\ldots, k$. Since ${\rm tr}[\rho_{ss}\mathbb W_{\mathsf D}(K)]=0$ for any $K$, the given form of the kernel follows. The range is clear from the definition, and the property $R_{\mathsf D}^2=R_{\mathsf D}$ is straightforward to check.
\end{proof}

\subsection{Markov covariance from a principal connection}\label{sec.unidentifiable.param}
We now proceed to show how the Markov covariance is naturally associated with a specific horizontal bundle for the principal $G$-bundle $\manifprim$, and we also define a Riemannian metric on the manifold $\manifprim$. In order to motivate this, we continue the the discussion from subsection \ref{subsec.QFI}. Indeed, the modified generator \eqref{eq.modified.generator} can be expressed as a fluctuation operator
$$
G^0_{\theta,a}(t) =  \sqrt t \, \mathbb F_t({\bf X}_{\mathsf D}(\dm_{\theta, a})),
$$ 
where we have used the suggestive notation ${\bf X}_{\mathsf D}$ for the \emph{real linear} isomorphism
\begin{eqnarray}
{\bf X}_{\mathsf D} : \mathcal{T}_{\mathsf D} &\to& M_{sa}(\mathbb C^d)\times M(\mathbb C^d)^{k} \\
\label{eq.mapX}
 \dot{\mathsf{D}} = (\dot H, \dot{L}^1, \ldots ,\dot{L}^k) &\mapsto & (\mathcal{E}_{\mathsf D}(\dot{\mathsf{D}}),  \dot{L}^1,\ldots , \dot{L}^k), \qquad
\label{eq.xdotd}
\end{eqnarray}
where $M(\mathbb C^d)^{k}$ is now considered as a real linear space with dimension $2kd^2$, while $M_{sa}(\mathbb C^d)$ is naturally a real linear space. Therefore, using the explicit expression provided in Proposition \ref{continuous_mc}, we obtain the following explicit expression of the QFI rate
\begin{eqnarray}
f^\theta_{a,b} &=& 4{\rm Re} \left(   {\bf X}_{\mathsf D} [\dot{\mathsf{D}}_{\theta,a}]\, ,\,  {\bf X}_{\mathsf D} [\dot{\mathsf{D}}_{\theta,b}] \right)_{\mathsf D}\nonumber\\
& =&
4\sum_{i=1}^k {\rm Re} \, \mathrm{tr}
\left[\rho_{ss}  
\left(\dot{L}^i_{\theta,a} - i[L^i_\theta , \mathbb{W}^{-1}\circ \mathcal{E}^0_{\mathsf D}(\dot{\mathsf{D}}_{\theta_a})  ]\right)^*\cdot
\left(\dot{L}^i_{\theta,b} - i[L^i_\theta , \mathbb{W}^{-1}\circ\mathcal{E}^0_{\mathsf D}(\dot{\mathsf{D}}_{\theta_b})  ]\right) \right].
\label{eq.qfi.rate2}
\end{eqnarray}
The QFI rate inherits the positivity property of the Markov covariance, but also the fact that it may not be positive \emph{definite}.

In conclusion, the \emph{real part} of the form
\begin{equation}\label{eq.inner.prod.ddot}
(\dot{\mathsf D},\dot{\mathsf{D}}')_{\mathsf D}:= \,({\bf X}_{\mathsf D}[\dot{\mathsf D}], {\bf X}_{\mathsf D}[\dot{\mathsf D}])_{\mathsf D}
\end{equation}
has a natural interpretation in terms of the output Fisher information. Its explicit form (see Proposition \ref{continuous_mc}) suggests the definition of the following projection on the tangent bundle over $\manifprim$:
$$
P:\mathcal T\to \mathcal T, \quad P_{\mathsf D}= {\bf X}_{\mathsf D}^{-1}\circ R_{\mathsf D}\circ {\bf X}_{\mathsf D}.
$$
Indeed, the bilinear form essentially depends on this projection:
$$
(\dot{\mathsf D},\dot{\mathsf{D}}')_{\mathsf D}=\sum_{i=1}^{k}
{\rm tr}\left[\rho_{ss}  [P_{\mathsf D}(\dot{\mathsf D})^{i}]^* \,P_{\mathsf D}(\dot{\mathsf D}')^{i} \right].
$$
In order to understand the intuitive meaning of $P_{\mathsf D}$, we now proceed to make a fundamental observation concerning the relation between the push-forward map $D_*$ defined in \eqref{eq.Dstar}, and the map $\mathcal E_{\mathsf D}$ defined in \eqref{eq.E} 
\begin{equation}\label{condexpect}
\mathcal E_{\mathsf D}\circ \mathsf D_*(X) = r\id +\mathbb W_{\mathsf D}(K), \quad X=(-iK,r).
\end{equation}
This relation can be verified by direct computation, and has appeared before in a different context \cite{AvFrGr12}.

Equation \eqref{condexpect}  is the key to defining the horizontal bundle for the identifiable parameters. Indeed, together with the formula of the projection $R_{\mathsf D}$ in Proposition \ref{continuous_mc}, this immediately gives
$$
P_{\mathsf D} = {\rm Id}-\mathsf D_*\circ \omega_{\mathsf D},
$$
where we have defined the following one-form from the tangent space to the Lie algebra $\mathfrak g$:
$$
\omega: \mathcal T_{\mathsf D}\to \mathfrak g, \quad \omega_{\mathsf D}(\dot{\mathsf D})=(\mathbb W_{\mathsf D}^{-1}\circ \mathcal E_{\mathsf D}^0(\dot{\mathsf D}), {\rm tr}[\rho_{\mathsf D}^{ss}\mathcal E_{\mathsf D}(\dot{\mathsf D})]).
$$
According to the identification of the Lie algebra as \eqref{liealg}, the range of this map fills the whole Lie algebra. Furthermore, we can easily verify the compatibility condition \eqref{comp}, and $G$-covariance \eqref{gcovariance} follows from the fact that
\begin{equation}\label{gcov}
\mathcal{E}_{g \mathsf D}(g_* (\dot D)) =  W^*\mathcal{E}_{\mathsf D}(\dot{\mathsf D})W, \quad \rho^{ss}_{g\mathsf D}=W^*\rho_{\mathsf D}^{ss}W, \quad g=(W,a)\in G,
\end{equation}
which is again straightforward to check. Hence, $\omega$ is the one-form of a unique principal connection on $\manifprim$, having $P_{\mathsf D}$ as its horizontal projection. Furthermore, the vertical subspaces are given by
$$\ker P_{\mathsf D}={\rm ran}\, \mathsf D_* = \nonident_{\mathsf D},$$
so this connection is compatible with the vertical bundle defining the "non-identifiable directions". We therefore achieve the aim described at the end of section \ref{sec.identifiability} by \emph{defining} the subspace of the identifiable directions to be the horizontal subspace:
$$
\ident_{\mathsf D} := {\rm ran}\, P_{\mathsf D} = \{ \dot{\mathsf D} \mid \mathcal E(\dot{\mathsf D})=0\}.
$$
The associated split $\mathcal T_{\mathsf D}=\nonident_{\mathsf D}\oplus \ident_{\mathsf D}$ now follows immediately from the general theory.
As a consequence of the $G$-invariance of the horizontal projection, the form $(\cdot , \cdot)_{\mathsf D} $ on $\mathcal{T}_{\mathsf D}$ is $G$-invariant in the sense that 
\begin{equation}\label{eq.inner.product.invariance}
(\dot{\mathsf D}, \dot{\mathsf D}^\prime)_{\mathsf D} = ( g_* \dot{\mathsf D}, g_* \dot{\mathsf D}^\prime )_{g\mathsf D} , \qquad
\dot{\mathsf D}, \dot{\mathsf D}^\prime\in \mathcal{T}^{id}_{\mathsf D}, g\in G.
\end{equation}
Hence, this form only depends on the equivalence class, so its real part determines a unique bilinear form on the base manifold $\mathcal P$. Moreover, it also only depends on the horizontal projection $P_{\mathsf D}(\dot{\mathsf D})$ of the tangent vectors; hence it becomes nondegenerate on the horizontal bundle, thereby defining a Riemannian metric on the base manifold $\mathcal P$.

We emphasise that the principal connection (together with the stationary state), completely determines the metric and the associated Fisher information. In this way the connection provides geometric insight on how the (in practice rather complicated) expression of the Fisher information arises; for a discussion on a classical analogy, see e.g. \cite{hanzon}. We demonstrate this in a concrete example in section \ref{sec.example} below.

\subsection{Symplectic structure and CCR-algebra for identification}\label{sec.ccr}

In Proposition \ref{continuous_mc} we defined the Markov covariance on the complex linear space $M(\mathbb C^d)^{k+1}$, and used the \emph{real linear} maps $${\bf X}_{\mathsf D}:\mathcal T_{\mathsf D}\to M(\mathbb C^d)^{k+1}$$
to induce an associated \emph{real} inner product $(\cdot, \cdot)_{\mathsf D}$ on the identifiable part of the tangent space, cf. equation \eqref{eq.inner.prod.ddot};  up to a constant factor, this inner product is the QFI rate. It is then natural to ask if the \emph{imaginary part} of the Markov covariance has any physical interpretation? We will show that the latter can be used to define an algebra of the canonical commutation relations (CCR) over the real space of identifiable parameters $\mathcal{T}^{id}_{\mathsf D}$, which will play the role of limit Gaussian model in the next section. 

On the real linear space $\ident_{\mathsf D}=\{\dot{\mathsf D}\mid \mathcal{E}_{\mathsf D}(\dot{\mathsf D})= 0\}=\ker P_{\mathsf D}$ we now define a \emph{complex structure} via
\begin{eqnarray}
\mathcal{J}_{\mathsf D}:\mathcal{T}^{id}_{\mathsf D} & \to & \mathcal{T}^{id}_{\mathsf D}\nonumber\\
\mathcal{J}_{\mathsf D}:  (\dot H, \dot L^1,\ldots,\dot L^k) &\mapsto & \left(\sum_{i=1}^k {\rm Re} \dot L^{i*}L^i\, ,\, i\dot L^1,\ldots,i\dot L^k\right).
\end{eqnarray}
Using the property that $\mathcal{E}_{\mathsf D}(\dot{\mathsf D})= 0$ for all vectors $\dot{\mathsf D}\in \mathcal{T}^{id}_{\mathsf D}$, 
it is easy to check that $\mathcal{J}_{\mathsf D}$ satisfies the defining property of a complex structure on $\mathcal{T}^{id}_{\mathsf D}$,  i.e. $\mathcal{J}_{\mathsf D}^2 =- {\rm Id}$. Furthermore, since $P_{\mathsf D}={\bf X}_{\mathsf D}^{-1}\circ R_{\mathsf D}\circ {\bf X}_{\mathsf D}$, we immediately see from Proposition \ref{projection} that
$$
{\bf X}[\mathcal{J}_{\mathsf D}(\dm)]=i{\bf X}[\dm]= (0, i\dot{L}^1, \ldots,  i\dot{L}^k), \qquad  
\dm \in \mathcal{T}^{id}_{\mathsf D},
$$
that is, the map ${\bf X}_{\mathsf D}$ is compatible with the natural complex structure of $M(\mathbb C^d)^{k+1}$. In fact, this is the \emph{only} way of defining a complex structure on $\mathcal{T}^{id}_{\mathsf D}$ in such a way that the restriction of ${\bf X}_{\mathsf D}$ to $ \mathcal{T}^{id}_{\mathsf D}$ is a \emph{complex} linear map.

When endowed with the complex structure $\mathcal J_{\mathsf D}$, the space $\ident_{\mathsf D}$ becomes a complex linear space; this is Hilbert space with respect to the inner product induced by the Markov covariance:
\begin{equation}\label{eq.inner.product.t.id}
(\dot{\mathsf D}, \dot{\mathsf D}^\prime)_{\mathsf D} := 
\left( {\bf X}_{\mathsf D}(\dot{\mathsf D}) \, ,\,  {\bf X}_{\mathsf D}(\dot{\mathsf D} ^\prime) \right)_{\mathsf D}, \qquad \dot{\mathsf D}, \dot{\mathsf D}^\prime\in \mathcal{T}^{id}_{\mathsf D}.
\end{equation}
The real part of this form gives the Riemannian metric and QFI rate on the real linear tangent space $\ident_{\mathsf D}$ as discussed above. In addition, the imaginary part can be used to construct a representation of the canonical commutation relations (CCR) over $ \mathcal{T}^{id}_{\mathsf D}$, together with a distinguished Fock state whose statistical interpretation is discussed in section \ref{sec.lan}.

\begin{definition}[{\bf CCR algebra for identifiable parameters}] \label{def.ccr}
Let $(\mathcal{T}^{id}_{\mathsf D}, \mathcal{J}_{\mathsf D}) $ and $(\dot{\mathsf D}, \dot{\mathsf D}^\prime)_{\mathsf D}$ be the complex linear space, and respectively inner product defined above.  On $\mathcal{T}^{id}_{\mathsf D}$ we define the symplectic form 
$$
\sigma^{\mathsf D}(\dm,\dm'):= {\rm Im} (\dot{\mathsf D}, \dot{\mathsf D}^\prime)_{\mathsf D} = \sum_{i=1}^k {\rm Im} \, {\rm tr}[\rho_{ss}\dot L_i^*\dot L_i']
$$
We define the CCR algebra $CCR(\mathcal{T}^{id}_{\mathsf D}, \sigma^{\mathsf D})$ generated by unitary Weyl operators $W(\dot{\mathsf D})$ with $\dot{\mathsf D}\in \mathcal{T}^{id}_{\mathsf D}$ satisfying the relations
$$
W(\dm)W(\dm')=e^{i\sigma^{\mathsf D}(\dm,\dm')}W(\dm+\dm'), \quad W(-\dm)=W(\dm)^*, 
$$
\end{definition}

On $CCR(\mathcal{T}^{id}_{\mathsf D}, \sigma^{\mathsf D})$ we define the Gaussian state $\varphi$ determined by the characteristic function
$$
\varphi(W(\dot{\mathsf D}) ) = e^{-\frac 18 f^{\mathsf D} (\dm,\dm)} 
$$
where 
\begin{equation}\label{eq.fd}
f^{\mathsf D} (\dm,\dm') := 4 {\rm Re} (\dm,\dm')_{\mathsf D} =  {\rm Re} \sum_i {\rm tr}[\rho_{ss}\dot L_i^*\dot L_i'].
\end{equation}
By a standard construction  \cite{Petz}, the CCR algebra can be represented on the Fock space $\mathcal{F}_{\mathsf D}$ over the Hilbert space 
$(\mathcal{T}^{id}_{\mathsf D}, \mathcal{J}_{\mathsf D}, (\dot{\mathsf D}, \dot{\mathsf D}^\prime)_{\mathsf D})$, in such a way that that 
$\varphi(W(\dot{\mathsf D})) = \langle 0 | W(\dot{\mathsf D}) |0\rangle$, where $ |0\rangle \in \mathcal{F}_{\mathsf D}$ is the vacuum state. For simplicity of notation we have identified the Weyl operators $W(\dot{\mathsf D})$ with their Fock representation.

\section{Examples of parametric models}\label{sec.example}
To illustrate the general theory we analyse several examples of one-parameter and multi-parameter models.
\subsection{One parameter models}
Let ${\mathsf D}_{\theta} : =(H, e^{-i\theta} L)$ be a one-parameter family, where we have chosen $m=1$ for simplicity. 
The corresponding one-dimensional tangent vector at ${\mathsf D}_{\theta =0}$ is
$
\dot{\mathsf D} = (0,iL ).
$
By applying equation \eqref{eq.derivative.psi} we find that the corresponding generator has differential equation
\begin{equation}\label{eq.generator.phase1}
dG_\theta(t) =  \sum_{i=1}^k \left[ j_t(L) dA^*(t) +  j_t(L^{*}) dA(t) + j_t(L^{*}L) dt \right]
\end{equation}
Then $\mathcal{E}^0_{\mathsf D} (\dot{\mathsf D}) =\left(- L^{*}L + \langle L^{*}L\rangle_{ss} \id \right)$, and 
${\bf X}_{\mathsf D}(\dot{\mathsf D}) = (\mathcal{E}^0_{\mathsf D} (\dot{\mathsf D}), iL )$. The QFI rate is 
$$
f^\theta= \mathrm{tr} \left[\rho_{ss}  \left(L + [L, \mathbb{W}^{-1} (L^*L -  \langle L^{*}L\rangle_{ss} \id )] \right)^2\right].
$$
Physically, this transformation can be implemented by placing a phase-shifter in each output channel, which gives each photon a phase shift $ e^{i\theta}$ \cite{MGGL}. This phase parameter is identifiable, and it is easy to see that 
$$
|\Psi^{\rm s+o}_\theta(t)\rangle  = \exp(-i\theta N(t))|\Psi^{\rm s+o}(t)\rangle 
$$
where $N(t)$ is the counting process associated to the number of photons up to time $t$ in the Bosonic environment. Equivalently, this can be written as 
$U^*(t) |\Psi^{\rm s+o}_\theta(t)\rangle = \exp(-i\theta N^{\rm out}(t)) |\phi\otimes \Omega\rangle $ where 
$N^{\rm out}(t) := U(t)^* N(t)U(t)$, is the output number of photons operator, whose differential form is 
\begin{equation}\label{eq.generator.phase2}
N^{\rm out}(t) =  dN(t) +  j_t(L) dA^*(t) +  j_t(L^{*}) dA(t) + j_t(L^{*}L) dt .
\end{equation}
By comparing \eqref{eq.generator.phase1} and \eqref{eq.generator.phase2} we see that the two generators are not identical. However, the difference is the term $dN(t)$ which annihilates the vacuum state, so the resulting \emph{action} of the generators \emph{is} identical. This illustrates that in general the generator is not unique but one can add terms which annihilate the vacuum, such as annihilation or number operator terms.

The second example we consider is that of the coupling constant, where $L_\theta = \theta L$, with unknown parameter 
$\theta\in \mathbb{R}$. The tangent vector is $\dot{\mathsf D} := (0, L)$, and $\mathcal{E}^0_{\mathsf D} (T) =0$. Therefore, ${\bf X}_{\mathsf D}(\dot{\mathsf D}) = (0, L )$ and the QFI rate is 
$$
f = \mathrm{tr} \left[\rho_{ss}  L^*L\right].
$$
which is simply the photon emission rate in the stationary regime.

In the third example we consider the model where the hamiltonian is known up to a multiplicative constant $H_\theta = \theta H$. The tangent vector is $\dot{\mathsf D} := (H, 0)$, and 
$\mathcal{E}^0_{\mathsf D} (T) =H - \langle H\rangle_{ss}$. 
Therefore, ${\bf X}_{\mathsf D}(\dot{\mathsf D}) = (H - \langle H\rangle_{ss} , 0 )$ and the QFI rate is 
$$
 \mathrm{tr} \left( \rho_{ss} [L , \mathbb{W}^{-1} (H - \langle H\rangle_{ss})  ]^*[L , \mathbb{W}^{-1} (H - \langle H\rangle_{ss})  ] \right).
$$

\subsection{Simplest multiparameter setting}\label{sec:multip_ex}

The geometric aspects are naturally trivial in a one-parameter model. In order to illustrate the full use of the theory developed above, we now consider the simplest nontrivial setting with $d=2$ and $k=1$, that is, $\manifprim$ is the open subset of $\{ (H, L)\mid H\in M_s(\mathbb C^2), \, L\in M(\mathbb C^2)\}$ consisting of ergodic dynamical parameters. The dimension of this manifold is $12$, and the number of identifiable parameters is $8$. Hence, full treatment of this simplest setting is still rather tedious, and we settle for looking at points on a physically relevant submanifold, extended suitably so as to allow for the full description of the relevant geometry. The model is the following $7$-dimensional submanifold:
\begin{align*}
H_{\Delta,\Omega,{\bf v}} &= \frac 12 \begin{pmatrix}\Delta & \Omega+v_1-iv_2 \\ \Omega+v_1+iv_2 & -\Delta +v_0 \end{pmatrix}, & L_{\alpha,\theta,{\bf v}} &= \alpha e^{i\theta} \begin{pmatrix}(iv_1-v_2)/\alpha^2 & 1+iv_0/\alpha^2 \\ 0 & (-iv_1+v_2)/\alpha^2 \end{pmatrix}.
\end{align*}
Here the three parameters ${\bf v}$ are auxiliary, and the rest have physical meaning at ${\bf v}=0$. In fact, we are looking at the off-resonant laser-driven two-level system with Rabi frequency $\Omega$ and detuning $\Delta$, in contact with a zero-temperature heat bath, with emission rate $\alpha^2$, and emitted photons monitored on the environment. In addition, we include the above discussed phase shift $\theta$ to the emitted photons. The auxiliary parameters are chosen such that their tangent vectors lie in the identifiable subspace at ${\bf v}=0$; their span is needed in order to describe the horizontal projections of the physical tangent vectors, as we will see below.

The quantum Fisher information associated with the three parameters $(\Delta,\Omega, \alpha)$ of this model has been compared with particular measurement strategies \cite{molmer}; we emphasise geometric aspects not discussed there, and have also included the phase parameter $\theta$. The main idea is to demonstrate how the rather complicated expressions of the Fisher information arise from considerably simpler geometric ingredients as a result of straightforward linear algebra. This provides insight on the structure of the physical system from the operational identification point of of view, and may eventually be useful in developing global estimation strategies in analogy to classical cases (see e.g. \cite{hanzon}).

\begin{figure}[h]
\begin{center}
\includegraphics[width=8cm]{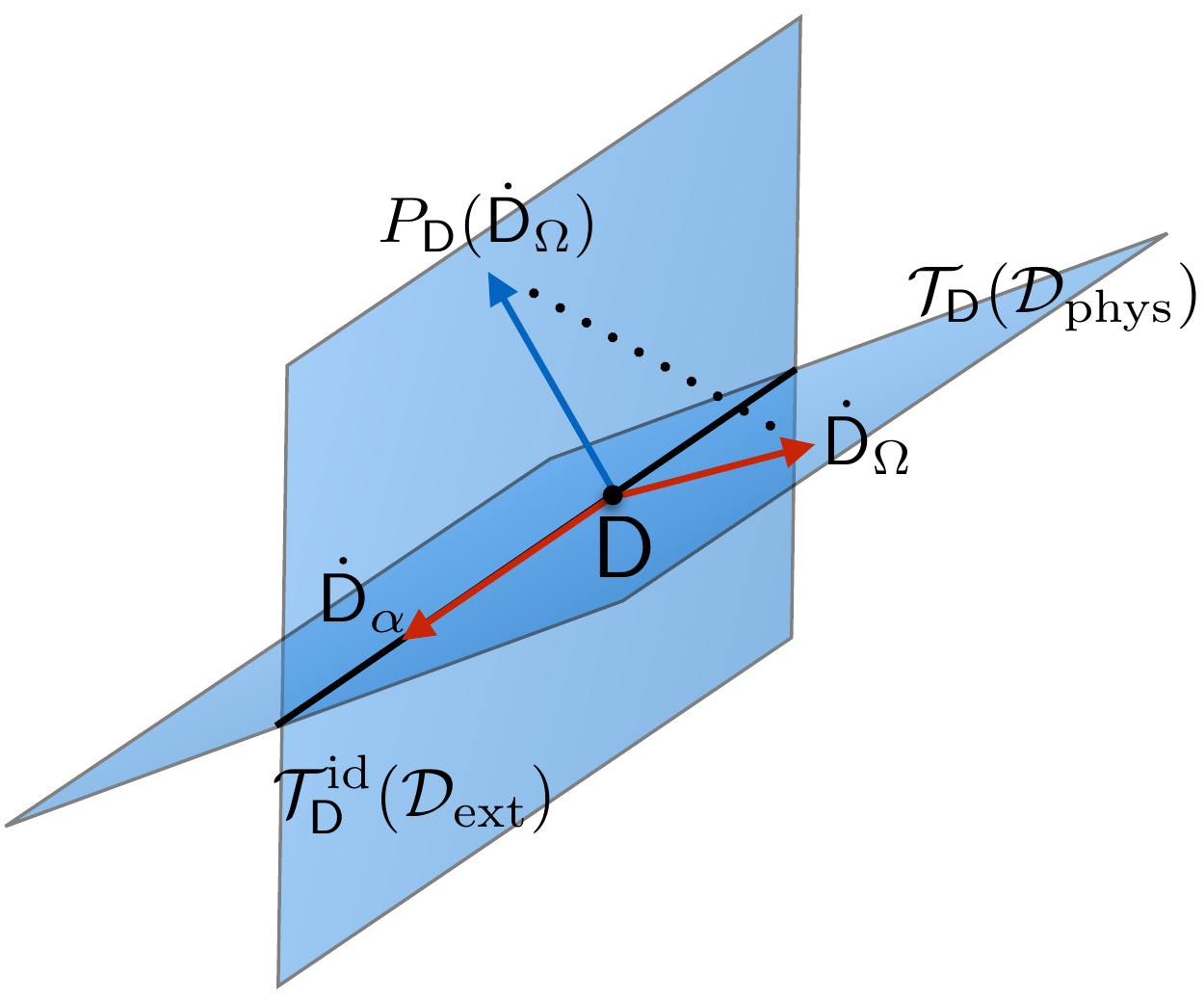}\hspace{2cm}
\end{center}
\caption{Sketch of the information geometry of two parameters in the driven two-level system. The tangent space of the physical manifold contains the directions of the decoupling parameter $\alpha$ and the driving frequency $\Omega$ (in red). Only the $\alpha$-direction lies in the identifiable subspace which supports the Fisher information metric; $\Omega$-direction needs to be projected there via the horizontal projection $P_{\mathsf D}$ of the principal connection.}
 \label{fig.geometry.example}
\end{figure}

Accordingly, we let $\manif_{\rm ext}$ denote the whole (extended) $7$-dimensional manifold, and $\manif_{\rm phys}=\{\mathsf D\in \manif_{\rm ext}\mid {\bf v}=0\}$ the physical submanifold. The dynamical parameters $\mathsf D\in \manif_{\rm phys}$ are ergodic except at special points; the unique stationary state is
$$
\rho_{ss}=\frac{\Omega}{\gamma}\left(
\begin{array}{cc}
 \gamma/\Omega-\Omega & \xi   \\
 \overline\xi  & \Omega \\
\end{array}
\right),
$$
where $\gamma=\alpha ^4+4 \Delta ^2+2 \Omega ^2$ and $\xi = 2 \Delta+i \alpha ^2$.

\subsubsection{The Lie algebra and unidentifiable directions} We begin the description of the geometry by finding the unidentifiable part of the tangent space of the physical manifold. We let $E_{ij}$ denote the natural basis matrices of $M(\mathbb C^2)$, so that e.g. $\sigma_z=E_{00}-E_{11}$. Note that matrices such as $E_{ij}$ and $iE_{ij}$ are linearly independent, since we look at the \emph{real linear} version of $M(\mathbb C^2)$. We parametrise the Lie algebra by
$$
\mathfrak g=\{X[{\bf w},r]\mid {\bf w}\in\mathbb R^3 \, ,r\in \mathbb R \},
$$
where $X[{\bf w},r]=(-iK_{\bf w},r)$, with $K_{\bf w}=f({\bf w})\id + {\bf w}\cdot \sigma$, and the ``gauge" function $f({\bf w})$ is irrelevant for the action of Lie algebra on the parameter manifold. As discussed above, this gauge can be fixed so that $K_{\bf w}$ has zero mean; this is convenient since the back-action given by the connection one-form will automatically have this gauge.
 The zero mean gauge for the Lie algebra is
$$f({\bf w})=-\gamma^{-1}\left(4\Delta  (w_1 \Omega +\Delta  w_3)-2 \alpha ^2 w_2 \Omega +\alpha ^4 w_3\right).$$
In order to describe the action of the Lie algebra on the tangent space, we compute explicitly the image of the push-forward $(\mathsf D_{\Delta,\Omega,\alpha,\theta})_*$, corresponding to the unidentifiable part of the tangent space. Omitting the subscript for simplicity, the tangent vectors induced by basic rotations and the unidentifiable phase are
\begin{align*}
\dot{\mathsf D}^{nonid}_x &:=\mathsf D_*(X[(1,0,0),0]) = (-\Delta \sigma_y, i\alpha e^{i\theta}\sigma_z) &
\dot{\mathsf D}^{nonid}_y &:=\mathsf D_*(X[(0,1,0),0]) = (\Delta \sigma_x-\Omega\sigma_z, -\alpha e^{i\theta} \sigma_z)\\
\dot{\mathsf D}^{nonid}_z &:=\mathsf D_*(X[(0,0,1),0]) = (\Omega \sigma_y, -2\alpha ie^{i\theta}E_{01}) & \dot{\mathsf D}^{nonid}_{phase} &:=\mathsf D_*(X[(0,0,0),1]) = (\id,0).
\end{align*}

\subsubsection{Symplectic structure of the identifiable subspace} The $8$-dimensional identifiable subspace supports the principal connection and Markov covariance. It is characterised by the condition $\mathcal E(\dot{{\mathsf D}})=0$. Using this condition, one can easily verify that the following vectors belong to this subspace: $(E_{11}, ie^{i\theta} E_{01}/\alpha)$, $(\sigma_x/2, ie^{i\theta}\sigma_z/\alpha)$,  $(\sigma_y/2, -e^{i\theta}\sigma_z/\alpha)$, and $(0, e^{i\theta}E_{01}/\alpha)$. Here the first three are exactly the tangent vectors of the auxiliary parameters ${\bf v}$. It turns out (see below) that the rest of the identifiable subspace is irrelevant for the physical model. The Markov covariance can be directly computed on this part of of the identifiable subspace; the corresponding matrix $M$ is given by $$M_{ij}:=(\dot D_{i},\dot D_{j})_{\mathsf D}={\rm tr}[\rho_{ss} \dot{L}_{i}^*\dot L_{j}],$$
where $j$ indexes the basis vectors. The basis is not symplectic, and hence the Markov covariance is not diagonal. There are several choices for a symplectic basis, obtainable from the above one via standard diagonalisation procedures. One fairly simple choice is
\begin{align*}
\dot D^q_1&=\left(0,\begin{pmatrix}0 & 1\\ 0 & 0\end{pmatrix}\right), & \dot D^p_1&=\frac{\alpha^2\gamma}{\Omega^2}\left(\begin{pmatrix}0 & 0\\ 0 & -1\end{pmatrix}, \frac{e^{i \theta}}{i \alpha}\begin{pmatrix}0 & 1\\ 0 & 0\end{pmatrix}\right)\\
\dot D^q_2&=\left(\begin{pmatrix}0 & -i\Omega/2\\ i\Omega/2 & \alpha^2\end{pmatrix}, \frac{e^{i \theta}}{\alpha}\begin{pmatrix}-\Omega & \xi\\ 0 & \Omega\end{pmatrix}\right), & \dot D^p_2&=\frac{\alpha^2\gamma}{2\Omega^4}\left(\begin{pmatrix}0 & \Omega/2\\ \Omega/2 & -2\Delta\end{pmatrix}, \frac{e^{i \theta}}{i \alpha}\begin{pmatrix}-\Omega & \xi\\ 0 & \Omega\end{pmatrix}\right),
\end{align*}
where the labels $q$ and $p$ refer to canonical coordinates. Indeed, in this basis, we have $M=F + i \Sigma$, where the symplectic form is the standard one,
and the Fisher information matrix is diagonal:
\begin{align*}
F &= \begin{pmatrix}\Omega^2/(\alpha^2\gamma) & 0\\ 0 & \alpha^2\gamma/\Omega^2\end{pmatrix}\oplus \begin{pmatrix}2\Omega^4/(\alpha^2\gamma) & 0\\ 0 & \alpha^2\gamma/(2\Omega^4)\end{pmatrix}, & \Sigma &= \begin{pmatrix}0 & -1\\ 1 & 0\end{pmatrix}\oplus \begin{pmatrix}0 & -1\\ 1 & 0\end{pmatrix}
\end{align*}
Here the pair $(\dot D^q_1,\dot D^p_1)$ corresponds to the first block, and $(\dot D^q_2,\dot D^p_2)$ to the second. 

\subsubsection{The connection form on the physical manifold} In order to investigate the geometry of the physical manifold, we first find the tangent space $\mathcal T(\manif_{\rm phys})\subset \mathcal T$, consisting of meaningful directions in the model. It is the span of the following tangent vectors:
\begin{align*}
\dot D_\Omega &= (\tfrac 12 \sigma_x,0), & \dot D_\Delta &= (\tfrac 12 \sigma_z, 0),\\
\dot D_{\alpha} &= (0,e^{i\theta}E_{01}), & \dot D_{\theta} &= \alpha (0,ie^{i\theta}E_{01}).
\end{align*}
These vectors span a $4$-dimensional subspace of $\mathcal T$; note that the dependence on the manifold point only comes with the phase parameter $\theta$. We can now determine the connection one-form on the physical manifold; this is a straightforward computation involving the inversion of the generator $\mathbb W$ on the zero-mean subspace. The result is
$$
\omega = \omega(\dot D_\Delta)d\Delta + \omega(\dot D_\Omega)d\Omega + \omega(\dot D_\alpha)d\alpha +\omega(\dot D_\theta)d\theta,
$$
where the components are given by
\begin{align*}
\omega(\dot D_\Delta) &= \gamma^{-1}X[-\alpha^{-2}(4 \Delta \Omega, 2\Omega\alpha^2, |\xi|^2),|\xi|^2/2] \\
\omega(\dot D_\Omega) &= 2\gamma^{-1}X[-\alpha^{-2}(\alpha^4+2\Omega^2,-2\Delta\alpha^2,2\Delta\Omega),\Delta\Omega]\\
\omega(\dot D_\alpha) &= 0 \\
\omega(\dot D_\theta) &= -\gamma^{-1}X[(4\Delta\Omega,2\alpha^2\Omega,|\xi|^2), \alpha^2\Omega^2].
\end{align*} 
Similarly, we could determine the connection on the extended manifold; however, the result is considerably more complicated and is not very illuminating. Using the above components together with the push-forward $\mathsf D_*$, we obtain the horizontal projection of the physical tangent space on the identifiable subspace; a direct computation shows that they are all in the above four-dimensional subspace, with components in the above chosen symplectic basis given by
\begin{align*}
P(\dot D_\Delta) &= \frac {1}{(\gamma\alpha)^2} \begin{pmatrix} -4 \alpha ^4\gamma \Delta \\ 2 \alpha ^4 \Omega ^2 \end{pmatrix}\oplus\begin{pmatrix} \alpha ^4\gamma \\ 4 \Delta  \Omega ^4 \end{pmatrix},
&P(\dot D_\Omega) &= \frac {1}{(\gamma\alpha)^2\Omega} \begin{pmatrix}-\alpha ^4\gamma \left(\gamma -8 \Delta ^2\right) \\-4 \alpha ^4 \Delta  \Omega^2\end{pmatrix}\oplus\begin{pmatrix}-2 \alpha ^4 \Delta\gamma\\ 2 \Omega ^4 \left(\alpha ^4+2 \Omega ^2\right)\end{pmatrix},\\
P(\dot D_\alpha) &= \alpha\begin{pmatrix} 1 \\ 0 \end{pmatrix}\oplus\begin{pmatrix} 0\\ 0\end{pmatrix}, &
P(\dot D_\theta) &= \frac{1}{\gamma^2} \begin{pmatrix}-4 \gamma\alpha ^4 \Delta\\ \Omega^2 \left(2 \alpha ^4-\gamma \right)\end{pmatrix}\oplus\begin{pmatrix} \alpha ^4\gamma\\ 4 \Delta  \Omega ^4\end{pmatrix}.
\end{align*}

\subsubsection{The Fisher information of the physical parameters} We can now easily compute the Fisher information by sandwiching the diagonal matrix $F$ with the above column vectors. The result is 
\begin{align*}
f_\Delta &= \frac{2 \Omega ^2 |\xi|^2\left(2 \alpha ^4+\Omega ^2\right)}{\alpha ^2 \gamma^3} & f_\Omega &= \frac{\alpha ^{12}+\alpha ^8 \left(8 \Delta ^2+6 \Omega ^2\right)+4 \alpha ^4 \left(4 \Delta ^4-2 \Delta ^2 \Omega ^2+3 \Omega ^4\right)+8 \Omega ^6}{\alpha ^2 \gamma^3}\\
f_{\alpha} &= \frac{\Omega^2}{\gamma} & f_{\theta} &=\frac{\alpha ^2 \Omega ^2 \left(-2 \Omega ^2 \left(\alpha ^4-12 \Delta ^2\right)+|\xi|^4+4 \Omega ^4\right)}{\gamma^3}.
\end{align*}
Note that the reason why these expressions are rather complicated is partially due to the fact that the physical directions do not lie in the identifiable subspace, but need to be projected there.
\subsubsection{The canonical coordinates of the physical parameters} This can be read off from the components of the above column vectors; for instance, at $\Delta =0$ (resonance) we have simply
\begin{align*}
W(P(\dot D_\Delta)) &= \exp \frac{i}{\gamma}\left(- 2\alpha^2\Omega^2\gamma^{-1} Q_1+\alpha^2P_1\right) &
W(P(\dot D_\Omega)) &= \exp \frac{i}{\gamma}\left( -\alpha^2\Omega^{-1}\gamma P_1+ 2\Omega^3\alpha^{-2}Q_2\right)\\
W(P(\dot D_\alpha)) &= \exp i\alpha P_1 & 
W(P(\dot D_\theta)) &= \exp \frac i \gamma \left(-\Omega^2(\alpha^4-2\Omega^2)\gamma^{-1} Q_1 +\alpha^4 P_2\right).
\end{align*}

Using the connection form, one could further investigate the global structure of the information geometry, in terms of the curvature, geodesics and parallel transport. This would be relevant for some of the future lines of research mentioned in the introduction, but beyond the scope of this paper. We only note that for instance the curvature two-form can easily be determined by straightforward although somewhat tedious computer algebra; this shows in particular that the connection is not flat, i.e. the horizontal bundle is not integrable.
\section{Local asymptotic normality in the multiparameter setting}\label{sec.lan}
In subsection \ref{subsec.QFI} we showed that the quantum Fisher information of the output state increases linearly in time as 
$F^\theta(t) \approx t f^\theta$, and we identified the QFI rate $f^\theta$ as the real part of the Markov covariance matrix of tangent vectors corresponding to changes in the parameter $\theta$, cf. equations \eqref{eq.qfi.rate}, \eqref{eq.qfi.rate2}. In this section we extend the statistical analysis by proving that the \emph{output state is asymptotically Gaussian} in the limit of large times, in a sense which will be defined precisely below. In effect this means that the output states for parameters in a local neighbourhood of a given dynamical parameter $\mathsf D_0$, can be approximated by a limit model which consists of a family of pure Gaussian states of the CCR algebra $CCR(\mathcal{T}^{id}_{\mathsf{D}_0}, \sigma^{\mathsf{D}_0})$ defined above, with mean determined by local changes in the unknown parameter, and covariance given by the QFI rate. Before stating the asymptotic normality result, we briefly review the general statistical concepts involved in its formulation. For more details about the general theory of quantum statistical models we refer to \cite{GutaJencova,GutaKiukas}.

\subsection{Convergence of quantum statistical models} \label{sec.statistical.models}
A \emph{quantum statistical model} over the parameter space 
$\Theta\subset\mathbb{R}^k$ is a family $\mathcal{Q} := \{\rho^\theta:\theta\in \Theta\}$ of quantum states on a fixed Hilbert space $\mathcal{H}$, which are indexed by an unknown parameter $\theta\in \Theta$. 
We are interested in characterising the asymptotic behaviour of an \emph{ordered set} of statistical models, in particular the convergence to a limit model. Such problems arise in quantum state estimation where the statistical models consist of ensembles of  identically prepared systems, and the order parameter is the size of the ensemble \cite{KahnGuta}, or in the estimation of dynamical parameters (system identification) where time plays the role of ``sample size". The latter case is the topic of this paper.

We start by noting that the space of statistical models is equipped with a natural notion of equivalence. Two models 
$\mathcal{Q}_1 := \{\rho_1^\theta:\theta\in \Theta\}$ and $\mathcal{Q}_2 := \{\rho_2^{\theta}:\theta\in \Theta\}$ 
(possibly on different Hilbert spaces) are statistically equivalent if there exist quantum channels $T,S$ 
between the appropriate state spaces such that 
\begin{equation}\label{eq.equivalence}
T(\rho^\theta_1) = \rho^{\theta}_2 ,\qquad S( \rho^{\theta}_2) = \rho^\theta_1
\end{equation}
for all $\theta\in \Theta$.
A consequence of the equivalence is that the probability distribution of any measurement $M$ on $\mathcal{Q}_1$ can be reproduced by a measurement on $\mathcal{Q}_2$ obtained by applying $S$ followed by $M$, and vice versa.  Therefore the two models have exactly the same optimal risks (figures of merit) for any statistical problem concerning the parameter $\theta$. In the special case when $\mathcal{Q}_1$ and $\mathcal{Q}_2$ are pure state models, it can be shown \cite{Chefles} that the models are equivalent if and only if there exist representative vectors (i.e. $\rho^\theta_1 = |\psi^\theta_1 \rangle\langle \psi^\theta_1| ,\rho^\theta_2 = |\psi^\theta_2 \rangle\langle \psi^\theta_2| $)  such that the overlaps of all pairs of vectors in the two models coincide
$$
\langle \psi^\theta_1 |  \psi^{\theta^\prime}_1\rangle=  \langle\psi^\theta_2 |  \psi^{\theta^\prime}_2\rangle,\qquad \theta,\theta^\prime\in \Theta.
$$ 
This shows that the intrinsic statistical properties of the model are encoded in the overlaps, up to an ambiguity in choosing the phases.

More generally, the theory of \emph{quantum sufficiency} \cite{PetzJencova} deals with the situation when the models are related by the channel transformation \eqref{eq.equivalence} only in one direction, so that one of the models is more informative that the other. However, such a relationship is still rather restrictive; in asymptotic statistics one is often interested in approximating a given model by a ``simpler" one which is ``close" to it in  a statistical sense. The above discussion suggests two ways of formalising this idea.  The first one is to define a notion of distance between models \cite{GutaJencova}, inspired by the classical theory developed by Le Cam \cite{LeCam}
\begin{definition}\label{def.strong.convergence}
Let $\mathcal{Q}_1$ and $\mathcal{Q}_2$ be two quantum statistical models over $\Theta$, defined as above. The deficiencies of one model with respect to the other are defined as
$$
\delta(\mathcal{Q}_1, \mathcal{Q}_2)  = 
\inf_T \sup_{\theta\in \Theta} \| T ( \rho_1^\theta) - \rho_2^\theta\|_1, \qquad
\delta(\mathcal{Q}_2, \mathcal{Q}_1)  = \inf_S \sup_{\theta\in \Theta} \| S ( \rho_2^\theta) - \rho_1^\theta\|_1,
$$
where the infima are taken over all quantum channels between the appropriated spaces, and the distance is given by the trace-norm  
$\|\tau{}\|_1 := {\rm Tr}(|\tau|)$. 
The Le Cam distance between the models $\mathcal{Q}_1$ and $\mathcal{Q}_2$ is defined as
$\Delta(\mathcal{Q}_1, \mathcal{Q}_2)= \max (\delta(\mathcal{Q}_1, \mathcal{Q}_2) ,\delta(\mathcal{Q}_2, \mathcal{Q}_1) )$.

A set of model $\mathcal{Q}_t:=\{\rho_t^\theta \,:\, \theta\in \Theta \}$ indexed by $t$ in $\mathbb{N}$ or $\mathbb{R}$  converges strongly (or in the sense of Le Cam) to a limit model $\mathcal Q:=\{\rho^\theta \,:\, \theta\in \Theta \}$ if $\Delta(\mathcal{Q}, \mathcal{Q}_t) \rightarrow 0$ as $t\rightarrow \infty$.
\end{definition}
It can be shown that two models are equivalent if and only if the Le Cam distance between them is zero. More generally, the Le Cam distance provides an upper bound to the the difference between optimal risks of statistical decision problems with bounded
loss functions \cite{GutaJencova}. Furthermore, the convergence to a simpler limit model can be used to identify asymptotically optimal measurement procedures for a given statistical decision problem, e.g. state estimation. This can be done by mapping the state $\rho^\theta_t$ through the channel $T_t$ onto the space of the limit model, followed by applying the optimal measurement for the limit model. An instance of this the phenomenon of local asymptotic normality for state estimation \cite{KahnGuta} which we illustrate below in the simplified setup of pure states. For this we formulate the second notion of convergence of models, based on the fidelity of the state vectors.

\begin{definition}[{\bf weak convergence of pure states statistical models}]
Let $\mathcal{Q}_t:= \{ \rho^\theta_t  :\theta\in \Theta\} $ be a set of pure states quantum statistical 
models on Hilbert spaces $\mathcal{H}_t$ over parameter space $\Theta\subset \mathbb{R}^k$, where the index $t$ is chosen from $\mathbb{N}$ or $\mathbb{R}$. The family $\mathcal{Q}_t$ is said to converge weakly to a model  $\mathcal{Q}:= \{ \rho^\theta :\theta\in \Theta\}$ on a Hilbert space $\mathcal{H}$, if there exists a choice of representative vectors (i.e. $\rho^\theta_t = |\psi^\theta_t \rangle\langle \psi^\theta_t| ,\rho^\theta_t = |\psi^\theta_t \rangle\langle \psi^\theta_t| $) such that   
$$
\lim_{t\to\infty} \langle \psi^\theta_t |  \psi^{\theta^\prime}_t\rangle  = 
 \langle \psi^\theta |  \psi^{\theta^\prime} \rangle ,\qquad \theta,\theta^\prime\in \Theta.
$$ 
\end{definition}
Given that each statistical model is completely determined by the overlaps of pairs of vectors with different parameters, the definition captures the intuitive idea that two models are ``close" to each other if they have similar overlaps. 
As a simple multidimensional example we consider the weak convergence of ensembles of identically prepared qubits to coherent states of a one mode continuous variables system, which is closely related to the theory of coherent spin states  \cite{Radcliffe}. Let  
$$
|\psi^{u}_n\rangle = \left[ \exp\left( \frac{i}{\sqrt{2n}}\left(u_y \sigma_x - u_x \sigma_y \right) \right) |0 \rangle \right]^{\otimes n}, \qquad u= (u_x,u_y)\in \mathbb{R}^2
$$
be a 2-dimensional family of i.i.d. qubit states obtained by rotating the basis vector $|0\rangle$ with generators given by the Pauli matrices $\sigma_x,\sigma_y$. Since the ensemble has size $n$, the statistical uncertainty in estimating rotation parameters is of the order of $n^{-1/2}$. It is then meaningful to restrict the attention to a shrinking region in the parameter space, and write the rotation parameters as $u/\sqrt{n}$ \cite{GutaKahn}. 
Due to the rescaling, the QFI of the ``local parameter" $u$ is a constant $2\times 2$ matrix 
$
f= 2 \id_2 
$ 
which plays a similar role to the QFI rate per unit of time defined in equations \eqref{eq.qfi.rate},\eqref{eq.qfi.rate2}.
We will now show that the sequence of local models $\mathcal{Q}_n = \{|\psi^u_n\rangle :u\in \mathbb{R}^2\}$
 converges weakly to the quantum Gaussian model $\mathcal{Q}= \{|u\rangle :u\in \mathbb{R}^2\}$, where $|u\rangle$ denotes the coherent state of a one mode continuous variables system with mean values for the canonical variables given by $\langle Q \rangle =u_x, \langle P\rangle =u_y$. Indeed, since 
$\langle 0 |\sigma_x|0\rangle = \langle 0 |\sigma_y|0\rangle =0$, by expanding in powers of $n^{-1/2}$ we obtain
$$
\lim_{n\to \infty} 
\langle\psi^{u}_n  |   \psi^{v}_{n}\rangle = 
\lim_{n\to \infty} 
\left( 
1 - \frac{1}{4n} \langle 0 |  (u_y \sigma_x - u_x \sigma_y)^2  | 0\rangle + o(n^{-1}) 
\right)^n
= \exp(\|u-v\|^2/4) = 
\langle  u|v \rangle, \qquad u,v\in \mathbb{R}^2.
$$
In particular, the limit model has QFI equal to $f =2 \id_2$ which is the inverse of the covariance of the vacuum state. 
Furthermore, one can show that the convergence holds also in the stronger sense of Le Cam, so that optimal estimation procedures for the limit Gaussian model can be ``pulled back" to asymptotically optimal measurements for the $n$ qubits ensemble. When the figure of merit  (or risk) is the mean square error $  \mathbb{E}( \|\hat{u} - u \|^2)$, the optimal measurement for estimating $u$ in the limit model is the heterodyne measurement; this can be seen as a noisy joint measurement of the canonical variables $Q$ and $P$ and it outcome $\hat{u}$ is an unbiased estimator of $u$ which has Gaussian distribution 
$N(u, \id)$. The variance of $\hat{u}$ can be written as $V= f^{-1} + \frac{1}{2}\id$ where the first term comes from the quantum covariance while the second is the minimum amount of ``noise" required for the simultaneous estimation of the means of the non-commuting observables $Q$ and $P$. Moreover, the estimator is normally distributed, which allows one to devise confidence regions for large $n$. By a Central Limit argument one can show that $Q$ and $P$ are the are the appropriately rescaled limits of the total spin observables $L_x$ and $L_y$ so that the optimal measurement is essentially a joint measurement of collective spin observables.

As we will see below, the key features of the i.i.d. qubit model are also present in the more complicated Markovian output setup, which we now proceed to consider.

\subsection{Multiparameter LAN for quantum Markov processes}
We start by considering a completely general model in which all identifiable parameters are unknown, and show how this model can be approximated locally by a Gaussian model on the CCR algebra of Definition \ref{def.ccr}.This result  can then be applied to the situation where some prior information is available and we deal with a lower dimensional model.

\subsubsection{Estimation of identifiable parameters.} We will consider that the physical dynamics is governed by an unknown dynamical parameter $\mathsf{D}$; however, since the latter cannot be completely identified from the stationary output state, we will focus on the estimation of \emph{all identifiable parameters} given by the equivalence classes $[\mathsf{D}]\in \mathcal{P}$. Similarly to the i.i.d. setup described in section \ref{sec.statistical.models}, we will be interested in the properties of the quantum output statistical model in the limit of large times.  It is then meaningful to consider parameters $[\mathsf{D}]$ lying in a shrinking neighbourhood of a fixed point 
$[\mathsf{D}_0]$ in $\mathcal{P}$, whose size is of the order of the statistical uncertainty $t^{-1/2}$. We will formulate two convergence results: the first one concerns the weak convergence of the system-output state, while the second deals with the strong convergence of the output state. Since the latter depends only on the equivalence class $[\mathsf{D}]$, the strong convergence can be formulated solely in terms of the parameter space $\mathcal{P}= \manifprim/G$. On the other hand, since the system-output state is \emph{not} invariant over equivalence classes, the weak convergence depends on the specific choice of dynamical parameters for each equivalence class. Geometrically, this choice is determined by a  \emph{section} of the principal bundle, i.e. a smooth map 
$s:\mathcal{P} \to \mathcal{D}^{\rm erg}$ such that $\pi\circ s ([{\mathsf D}]) = [{\mathsf D}]$ for $[{\mathsf D}]$ in a local neighbourhood of $[{\mathsf D}_0]$. We will assume that $s$ is  ``horisontal" in the sense that the tangent space to 
$s(\mathcal{P})$ at ${\mathsf D}_0$ is the horisontal space $\mathcal{T}^{id}_{{\mathsf D}_0}$. The intuition here is that the 
changes along equivalence classes of dynamical parameters are not observable in the output state, while those along tangent vectors in $\mathcal{T}_{\mathsf{D}_0}^{id}$ describe all the identifiable parameters.
Although the theory can be developed in a coordinate-free way, for concreteness we consider a local coordinates chart in a neighbourhood of $[\mathsf{D}_0]$ defined by
$$
\mathcal{C}: \mathcal{P} \to \mathcal{O}\subset \mathbb{R}^{\delta^{id}} 
$$
where $ \mathcal{O}$ is a open ball centred at the origin, and $ \mathcal{C}([\mathsf{D}_0]) = 0$. For simplicity we denote the parameter with coordinate $u$ by $[\mathsf{D}]_u$ and the corresponding ``lifted" dynamical parameter  by 
$\mathsf{D}_u:= s([\mathsf{D}]_u )$. The tangent vectors 
$$
[\dot{D}]_a  :=\left.\frac{ \partial [\mathsf{D}]_u}{ \partial u_a}\right|_{u=0} , \qquad 
\dot{D}_a :=\left.\frac{ \partial \mathsf{D}_u}{ \partial u_a}\right|_{u=0}
\qquad  a=1 ,\ldots ,\delta^{id}
$$
form a basis of the space $\mathcal{T}_{[\mathsf{D}_0]} $, and respectively $\mathcal{T}^{id}_{{\mathsf D}_0}$.
With these notations we define two local statistical models corresponding to the system-output state and respectively the output state at time $t$. 
\begin{definition}\label{def.models.lan}
Let $s,\mathcal{C}, [\mathsf{D}]_u\in \mathcal{P}, \mathsf{D}_u\in \mathcal{D}^{\rm erg} $ 
be define as above with coordinate $u\in\mathcal{O}\subset \mathbb{R}^{\delta^{id}} $ in a  neighbourhood of the origin. 
The quantum statistical models of system-output state and respectively the output state at time $t$ are defined by 
$$
\mathcal{Q}_t:= \left\{ \left|\Psi^{\rm s+o}_{u/\sqrt{t}}(t)\right\rangle :  u  \in\mathcal{O}\subset \mathbb{R}^{\delta^{id}}   \right\},\quad
\tilde{\mathcal{Q}}_t:= \left\{ \rho^{\rm out}_{u/\sqrt{t}}(t) :  u \in \mathcal{O}\subset \mathbb{R}^{\delta^{id}} \right\}.
$$
with dynamics generated by $\mathsf{D}_{u/\sqrt{t}}$. Furthermore, we define the (pure states) Gaussian model 
$$
\mathcal{G}:= \left\{  \rho_u := |u\rangle \langle u | :   u\in \mathcal{O}\subset  \mathbb{R}^{\delta^{id}}  \right\} 
$$
where $|u\rangle =W( \sum_a u_a \dot{\mathsf{D}}_a )|\Omega\rangle$ is the coherent state of the CCR algebra 
$CCR(\mathcal{T}^{id}_{\mathsf D_0}, \sigma^{\mathsf D_0})$, cf. Definition \ref{def.ccr}.
\end{definition}
The overlaps of the coherent states $|u\rangle$ can be computed from 
Definition \ref{def.ccr} and are given by
\begin{equation}\label{eq.overlap.coherent}
\langle u|u^\prime\rangle= 
\exp\left( - \frac{1}{8} (u-u^\prime)^T f^{\mathsf{D}_0} (u-u^\prime) + i u^T \sigma^{\mathsf{D}_0}  u^\prime   \right). 
\end{equation}
From this one can deduce that the Gaussian model $\mathcal{G}$ has quantum Fisher information $f^{\mathsf{D}_0}$, equal to the QFI rate of the system-output model $\mathcal{Q}_t$. 
The following theorem explains this connection by showing that the system-output and respectively output local models converge the  to the Gaussian limit model. From the practical viewpoint, this means that the linear QFI scaling with rate $f^{\mathsf{D}_0}$ is \emph{asymptotically achievable}, and moreover, the optimal measurement has asymptotically Gaussian distribution, cf.  \cite{KahnGuta} for a detailed discussion of the interpretation of local asymptotic normality. 
\begin{theorem}[{\bf local asymptotic normality}]\label{th.lan}
Let $\mathcal{Q}_t, \tilde{\mathcal{Q}}_t, \mathcal{G}$ be the system-output, output, and Gaussian models introduced in Definition \ref{def.models.lan}.  The following statements hold.

1) The pure states models $\mathcal{Q}_t$ converges weakly to the Gaussian model $\mathcal{G}$. More precisely, there exists a particular choice of the (unphysical) phase angle $\phi(u)$ of the coherent state $|u\rangle$ such that
\begin{equation}\label{eq.weak.lan.phase}
\lim_{t\to\infty} 
\left.\left\langle \Psi^{\rm s+o}_{u/\sqrt{t}}(t) \right| \Psi^{\rm s+o}_{u^\prime/\sqrt{t}}(t) \right\rangle = e^{i \phi(u')-i \phi(u)}
\langle u |u^\prime \rangle, \qquad u,u^\prime\in \mathcal{O}\subset \mathbb{R}^{\delta^{id}} . 
\end{equation}

2) The mixed states models $\tilde{\mathcal{Q}}_t$ converge strongly to the the Gaussian model $\mathcal{G}$, i.e. 
$\Delta(\tilde{\mathcal{Q}}_t, \mathcal{G})\to 0$. More precisely, there exist quantum channels $T_t, S_t$ such that
\begin{eqnarray*}
&&\lim_{t\to\infty} \sup_{u \in \mathcal{O}} \left\| T_t \left( \rho^{\rm out}_{u/\sqrt{t}} (t) \right) - \rho_u \right\|_1 =0\\
&&\lim_{t\to\infty} \sup_{u \in \mathcal{O}} \left\| S_t \left(\rho_u\right) -  \rho^{\rm out}_{u/\sqrt{t}} (t) \right\|_1 =0.
\end{eqnarray*}

\end{theorem}

In the reminder of this section we give the main idea of the proof and discuss the physical interpretation. The technical details can be found in the Appendix. Recall that the system-output state is given by $|\Psi^{\rm s+o}_{\mathsf{D}}(t)\rangle=  U_{\mathsf{D}} (t) |\varphi\rangle\otimes |\Omega\rangle$ where $U_\mathsf{D} (t)$ is the unitary defined by the QSDE \eqref{QSDE}. By using Ito calculus it can be shown \cite{CatanaGutaBouten} that the overlaps of system-output states for different dynamical parameters can be expressed in terms of a contractive (non-CP) semigroup
$$
\left.\left\langle \Psi^{\rm s+o}_{\mathsf{D}}(t) \right| \Psi^{\rm s+o}_{\mathsf{D}^\prime }(t)\right\rangle 
=
\left. \left\langle \varphi \right| 
e^{ t \mathbb{W}_{\mathsf{D}  ,\mathsf{D}^\prime }} (\id) 
\left|\varphi \right\rangle \right. .
$$
where $W_{\mathsf{D}, \mathsf{D}^\prime}$ is the ``off-diagonal" semigroup generator
$$
\mathbb W_{\mathsf{D}, \mathsf{D}^\prime}(X)=
i(H X  - X H^\prime)  +  \sum_i  \left[ 
L_{i}^* X L^\prime_{i}
-\frac 12   \left( L_{i}^*L_{i} X + X L^{\prime*}_{i} L^\prime_{i}) \right)
\right], \qquad \mathsf{D} = (H, {\bf L}), \quad \mathsf{D}^\prime = (H^\prime, \mathbf{L}^\prime)
$$
which coincides with the usual Markov generator $\mathbb{W}_\mathsf{D}$ for $\mathsf{D}= \mathsf{D}^\prime$. When choosing $\mathsf{D}= \mathsf{D}_{u/\sqrt{t}}$ 
and $\mathsf{D}^\prime= \mathsf{D}_{u^\prime /\sqrt{t}}$  the generator can be expanded as
$$
\mathbb{W}_{\mathsf{D}, \mathsf{D}^\prime} = \mathbb{W}_{\mathsf{D}_0} + \frac{1}{\sqrt{t}} 
\mathbb L_1[u,u']+ \frac{1}{2t} \mathbb L_2[u,u'] +O(t^{3/2}).
$$
Using a version of the Trotter-Kato second order perturbation Theorem for semigroups (cf. Theorem 2.2 in \cite{CatanaGutaBouten}) one can show that \eqref{eq.weak.lan.phase} holds with an explicit choice of the phase angle $\phi(u)$ as a quadratic form in $u$. The details of the calculations can be found in the Appendix. Note that since the phase $e^{i\phi(u)}$ is unphysical, it could have been incorporated in the definition of the coherent state $|u\rangle$, or in that of the system-output state $|\Psi^{\rm s+o}_{u/\sqrt{t}}(t) \rangle$.  

The second part of the Theorem can be proven by following the lines of an analogous discrete-time result, cf. Theorem 7 in \cite{GutaKiukas}. The main ideas are as follows. Let $\rho_{ss}= \sum_{m}\Lambda_{m} |e_{m}\rangle\langle e_{m}|$ be the spectral decomposition of the stationary state for some dynamical parameter $\mathsf{D}$. The stationary output state is given by 
$$
\rho^{\rm out}_{\mathsf{D}}(t)=\sum_{m,m'} \Lambda_{m} |\psi_{mm'}(t) \rangle\langle \psi_{mm'}(t)|,
$$
where, up to normalisation, $|\psi_{mm'}(t) \rangle$ are the conditional output states obtained by initialising the system in state $|e_m\rangle$ and projecting on state $|e_{m'}\rangle$ at time $t$, cf. proof of Lemma \ref{secondlemma} in Appendix.  For large times, the overlaps of the different  pure components $|\psi_{mm'}(t) \rangle$ vanish exponential fast; more generally, if $\mathsf{D}= \mathsf{D}_{u/\sqrt{t}}$ and 
$\mathsf{D}^\prime =\mathsf{D}_{u^\prime/\sqrt{t}}$ are two dynamical parameters in the local neighbourhood of $\mathsf{D}_0$ (i.e. $u,u^\prime \in \mathcal{O}$) then all the overlaps of components with \emph{different} indices decay exponentially uniformly in 
$u,u'$. This can be shown by expressing the overlaps in terms of the deformed generator $\mathbb{W}_{\mathsf{D},\mathsf{D}'} $
$$
\langle \psi_{mm'}(t) |\psi_{nn'}(t) \rangle = 
\left\langle e_{m'} | e^{t\mathbb{W}_{\mathsf{D},\mathsf{D}'}} (|e_n\rangle \langle e_m|) | e_{n'}\right\rangle
$$
and following the steps of the proof of Theorem 3 in \cite{GutaKiukas}, in particular the argument following equation (35). This implies that the components can be distinguished with vanishing error probability, without the knowledge of the local parameter $u$. Each pure component satisfies the weak version of the local asymptotic normality, which can be upgraded to the strong version as in Theorem 7 of \cite{GutaKiukas}, which in turn employs a general result  described in Lemma 5.  
Combining this with the fact that the pure components can be distinguished allows to construct the channels $T_t, S_t$ as in \cite{GutaKiukas}.

\subsubsection{Estimation for specific model of dynamical parameters} In the previous subsection we considered the problem of estimating all identifiable parameters, and showed how this becomes a quantum Gaussian estimation problem. 
Here, we show how this general result can be used for estimating an unknown parameter of the dynamics. Suppose that the the dynamical parameter $\mathsf{D}$ is known to depend on $\theta\in \mathbb{R}^m$ as described in section \ref{subsec.QFI}, so that 
$  \mathsf{D}=\mathsf{D}_\theta$. Let $\theta_0$ be a fixed but arbitrary parameter value and let 
$$
\dot{\mathsf{D}}_{a} :=\left (
\left. \frac{\partial H}{\partial \theta_a} \right|_{\theta_0}  , 
\left.\frac{\partial L^1}{\partial \theta_a}\right|_{\theta_0},
\ldots,
\left.\frac{\partial L^k}{\partial \theta_a}\right|_{\theta_0} 
\right) = 
(\dot{H}_{a}, \dot{L}^1_{a}, \dots, \dot{L}^k_{a}) \in \mathcal{T}_{\mathsf{D}_{\theta_0}}
$$
be the tangent vectors associated to the different directions in the parameter space $\mathbb{R}^m$.

 The stationary output's QFI rate matrix $f$ at a given point $\theta_0$ can be computed using the explicit formula \eqref{eq.qfi.rate2},  and we assume that $\theta$ is identifiable so that $f$ is a strictly positive matrix. We consider a local parametrisation around 
 $\theta_0$ given by $\theta= \theta_0 +h/\sqrt{t}$, with local parameter $h\in \mathbb{R}^m$. Since the stationary state depends only on the equivalence class $[\mathsf{D}]$, the statistical model can be projected onto the base space $\mathcal{P}$ giving rise to a local model $[D]_{h/\sqrt{t}}$, with $h\in\mathcal{O}^\prime\subset\mathbb{R}^m$, which can be seen as sub-model of the `full' model considered in the previous subsection. In particular, the asymptotic normality Theorem \ref{th.lan} applies directly to the sub-model. However, in general it may happen that the ``full'' Gaussian limit model may be ``too large'', and one can use a restricted model defined as follows. Recall that $\mathcal{T}^{id}_{\mathsf D_{\theta_0}}$ 
 is a Hilbert space  with inner product \eqref{eq.inner.product.t.id}, which defines the 
 CCR algebra $CCR(\mathcal{T}^{id}_{\mathsf D_{\theta_0}}, \sigma^{\mathsf D_{\theta_0}})$ and the Gaussian state $|0\rangle$. 
 Let $P(\dot{\mathsf{D}}_{a})$ be the projection of the tangent vector $\dot{\mathsf{D}}_{a}$ onto $\mathcal{T}^{id}_{\mathsf D_{\theta_0}}$, and define $\mathcal{T}^{\prime}$ to be the (complex) subspace spanned by these projections, with $a=1,\dots , m$. The subspace defines a CCR subalgebra $CCR(\mathcal{T}^\prime, \sigma^{\mathsf D_{\theta_0}})$, and the restriction of the Fock state $|0\rangle$ to this subalgebra is also a Fock state which we denote by the same symbol. 
 
As a concrete example, consider the driven two-level model of Section \ref{sec:multip_ex}, where the total identifiable subspace is 8-dimensional, while the subspace spanned by the projections of the physical tangent vectors is four-dimensional, associated with the CCR-algebra of four canonical quadratures. 

We now obtain the following asymptotic normality result for the model $\mathsf{D}_\theta$ in the neighbourhood of $\theta_0$.
 
\begin{corollary}
Let 
$$
\mathcal{Q}^\prime_t:= 
\left\{ \rho^{\rm out}_{h/\sqrt{t}}(t) :  h \in \mathcal{O}^\prime\subset \mathbb{R}^{m} \right\}, \qquad
\mathcal{G}^\prime := 
 \left\{  \rho^\prime_h := |h\rangle \langle h | :   h\in \mathcal{O}^\prime \subset  \mathbb{R}^{m}  \right\} 
$$
denote the local quantum statistical model of the output state associated to the dynamical parameter 
$\mathsf{D}_{\theta_0 + h/\sqrt{t}}$, and respectively the Gaussian model associated to the algebra 
$CCR(\mathcal{T}^\prime, \sigma^{\mathsf D_{\theta_0}})$, with $|h\rangle := W(h)|0\rangle$. Then $\mathcal{Q}^\prime_t$ 
converge strongly to the the Gaussian model $\mathcal{G}^\prime$, i.e. 
$\Delta(\mathcal{Q}^\prime_t, \mathcal{G}^\prime)\to 0$. More precisely, there exist quantum channels $T^\prime_t, S^\prime_t$ such that
\begin{eqnarray*}
&&\lim_{t\to\infty} \sup_{h \in \mathcal{O}^\prime} \left\| T^\prime_t \left( \rho^{\rm out}_{h/\sqrt{t}} (t) \right) - \rho^\prime_h \right\|_1 =0\\
&&\lim_{t\to\infty} \sup_{h \in \mathcal{O}^\prime} \left\| S^\prime_t \left(\rho^\prime_h\right) -  \rho^{\rm out}_{h/\sqrt{t}} (t) \right\|_1 =0.
\end{eqnarray*}

\end{corollary}

\noindent {\bf Acknowledgment.} This work was supported by the EPSRC project EP/J009776/1. JK also acknowledges additional support from the EPSRC project EP/M01634X/1.
We thank K. Macieszczak, J. P. Garrahan and I. Lesanovsky for useful discussions during the preparation of this work.

\section{Appendix: proofs}\label{proof.identifiability.ct}
\subsection{Proof Theorem \ref{equiv_thm}} The proof of Theorem \ref{equiv_thm} is inspired by a related argument from \cite{BaumgartnerNarnhofer}, and requires some auxiliary lemmas.

\begin{lemma}\label{firstlemma} Let ${\mathsf D}_l := (H_{l}, L^1_{l}, \ldots , L^k_l )$, $l=1,2$ be two dynamical parameters with system spaces $\mathcal{H}_l$ and assume that both dynamics are ergodic. We define the maps
\begin{align*}
\mathbb W_{ll'}:&\mathcal B(\hi_{l'},\hi_{l})\to \mathcal B(\hi_{l'},\hi_{l}), & \mathbb W_{ll'}(X)&=-iXH_{l',\rm eff}+iH_{l,\rm eff}^*X+\sum_{i=1}^{k} L^{i*}_{l} X L^i_{l'}.
\end{align*}
for $l,l'=1,2$. Then the following conditions are equivalent:
\begin{itemize}
\item[(i)] $\mathbb W_{12}$ has a purely imaginary eigenvalue;
\item[(ii)] $\mathbb W_{21}$ has an purely imaginary eigenvalue;
\item[(iii)] there exists a unitary operator $U:\hi_2\to\hi_1$, and $r\in \mathbb R$, such that 
$L^i_2 = U^* L^i_1 U$ for all $i=1,\dots, k$, and $H_2 = U^* H_1 U - r\id$.
\end{itemize}
If any of these conditions hold, then $\mathbb W_{12}(U)=irU$ and $\mathbb W_{21}(U^*)=-irU^*$.
\end{lemma}
\begin{proof}
Conditions (i) and (ii) are clearly equivalent: $\mathbb W_{ll'}(F)=irF$ with some $r\in\mathbb R$ and $F\in \mathcal B(\hi_{l'},\hi_{l})$, then $\mathbb W_{l'l}(F^*)=-irF^*$.
Assuming (iii) we have
$$
\mathbb W_{12}(U)= -i UH_{2,\rm eff}+iH_{1,\rm eff}^*U+\sum_{i=1}^{k} L^{i*}_{1} U L^i_{2}
=U\mathbb W_{22}(\id)+irU=irU,
$$
i.e. (i) holds, with $ir$ the corresponding eigenvalue. Thus, the only nontrivial implication is (i)$\implies$ (iii).

Let us define the families of isometries $V_l(t): \mathcal{H}\to \mathcal{H}_l\otimes \mathcal{F}$ such that 
$V_l(t) |\varphi\rangle = U_l(t) |\varphi\rangle\otimes \Omega\rangle$, with $U_l(t)$ the unitary generated by the dynamical parameter ${\mathsf D}_l$, cf. equation \eqref{QSDE}. Then $T_{ll',t}(X):=V_l(t)^*(X\otimes \id_{\mathcal F})V_{l'}(t)=e^{it\mathbb W_{ll'}}(X)$. Assume now (i), and let $r$ and $F$ be such that $\mathbb W_{12}(F)=irF$. Then $T_{12,t}(F)=e^{itr}F$, and since $\mathbb W_{21}(F^*)=-irF^*$, we also have $T_{21,t}(F^*)=e^{-irt}F^*$. For each $t$ we have $V_1(t)V_1(t)^*\leq \id_{\hi_1\otimes\mathcal F}$, so
\begin{align*}
T_{22,t}(F^*F)&=V_2^*(t)(F^*\otimes \id_{\mathcal F})(F\otimes \id_{\mathcal F})V_2(t)\\
&\geq V_2^*(t)(F^*\otimes \id_{\mathcal F})V_1(t)V_1(t)^*(F\otimes \id_{\mathcal F})V_2(t)=T_{21,t}(F^*)T_{12,t}(F) =F^*F.
\end{align*}
Let $P$ be the projection onto the eigenspace of $F^*F$ corresponding to its largest eigenvalue $\|F^*F\|$. Now $\lim_{t\rightarrow\infty} T_{22,t}(X) ={\rm tr}[\rho_{ss,2} X]\id_{\hi_2}$ by ergodicity, so
$$
{\rm tr}[\rho_{ss,2}F^*F]=\lim_{t\rightarrow\infty} {\rm tr}[P]^{-1}{\rm tr}[PT_{22,t}(F^*F)]\geq {\rm tr}[P]^{-1}{\rm tr}[PF^*F]=\|F^*F\|.
$$
This implies that ${\rm tr}[\rho_{ss,2}F^*F]=\|F^*F\|$, i.e. $\rho_{ss,2}$ is supported in the projection $P$. But $\rho_{ss,2}$ has full rank in $\hi_2$, so $P=\id_{\hi_2}$, and, consequently, $F^*F=\|F^*F\|\id_{\hi_2}$. By proceeding in exactly the same way starting from $T_{11,t}$, we show that $FF^*=\|FF^*\|\id_{\hi_1}$. Denote $\alpha:= \|FF^*\|=\|F^*F\|$, and $U:=\alpha^{-\frac 12} F$. Then $U:\hi_2\to\hi_1$ is a unitary operator between the two Hilbert spaces and in particular, $\dim\hi_1=\dim\hi_2$. Moreover, we now have
\begin{equation}\label{eigenv}
ir\id_{\mathcal H_2}=U^*\mathbb W_{12}(U)=-i(H_2-U^*H_1U)-\frac 12\sum_i(L^{i*}_{2} L^i_{2}+U^*L^{i*}_{1} L^i_{1}U)+\sum_{i} (U^*L^i_{1}U)^*L^i_{2};
\end{equation}
taking the real part of the trace of this equation gives
\begin{equation}\label{realpart}
{\rm Re}\sum_{i} {\rm tr}[(U^*L^i_{1}U)^*L_{2,i}]=\frac 12\sum_{i} {\rm tr}[L^{i*}_{2} L^i_{2}+U^*L^{i*}_{1} L^i_{1}U]
\end{equation}
which implies that the (generally valid) inequalities
\begin{align*}
{\rm Re}\sum_{i} {\rm tr}[(U^*L^i_{1}U)^*L^i_{2}]&\leq \left|\sum_{i} {\rm tr}[(U^*L^i_{1}U)^*L^i_{2}]\right|\leq \sqrt{\sum_{i} {\rm tr}[L^{i*}_{2}L^i_{2}]\sum_i {\rm tr}[(U^*L^i_{1}U)^*(U^*L^i_{1}U)]}\\
&\leq \frac 12 \sum_{i} {\rm tr}[L^{i*}_{2} L^i_{2}+U^*L^{i*}_{1} L^i_{1}U]
\end{align*}
are in fact equalities. In particular,
\begin{equation}\label{imag}
{\rm Im} \sum_{i} {\rm tr}[(U^*L^i_{1}U)^*L_{2,i}]=0.
\end{equation}
Moreover, since the second inequality is Cauchy-Schwartz for the scalar product $\sum_i {\rm tr}[A_i^*B_i]$ of $k$-tuples $(A_1,\ldots, A_{k})$ of Hilbert-Schmidt operators, it follows that there exists a scalar $c\in \mathbb C$ such that 
$U^*L^i_{1}U=cL^i_{2}$ for all $i$. Putting this into \eqref{imag} we see that $c\in \mathbb R$, and from \eqref{realpart} it follows that $c=1$. Finally, from \eqref{eigenv} we then get $r=-H_2+U^*H_1U$, which proves (iii).
\end{proof}

For reader's convenience we formulate the following simple lemma using the notations of the input-output setting, but the statement holds in a general context.
\begin{lemma}\label{secondlemma} Let $\mathcal F$, and $\mathcal H_l$, $l=1,2$, be Hilbert spaces. For each $t\geq 0$ let $$V_l(t):\hi_l\to\hi_l\otimes\mathcal F, \quad l=1,2,$$ be an isometry, and define the maps $$T_{ll',t}(X):=V_l(t)^*(X\otimes \id_{\mathcal F})V_{l'}(t), \quad X\in \mathcal B(\mathcal H_{l'},\mathcal H_l).$$ Suppose that
\begin{align*}
\lim_{t\rightarrow\infty} T_{ll',t}(\cdot)&=\delta_{ll'} {\rm tr}[\rho_{ss,l}(\cdot)]\id_{\mathcal H_l}
\end{align*}
holds for some states $\rho_{ss,l}$, and define $\rho_{l}^{\rm out}(t):={\rm tr}_{\mathcal H_l}[V_l(t)\rho_{ss,l}V_l(t)^*]$. Then 
$$
\lim_{t\rightarrow\infty} {\rm tr}[\rho_{1}^{\rm out}(t)^2],\quad {\rm and}\quad \lim_{t\rightarrow\infty} {\rm tr}[\rho_{2}^{\rm out}(t)^2]
$$
exist and are strictly positive, while
\begin{equation}\label{asymptotic_ortho}
\lim_{t\rightarrow\infty} {\rm tr}[\rho_{1}^{\rm out}(t)\rho_{2}^{\rm out}(t)]=0.
\end{equation}
\end{lemma}
\begin{proof}
We write $\rho_{ss,l} = \sum_{m}\Lambda_{l,m} |e_{l,m}\rangle\langle e_{l,m}|$ where $\Lambda_{l,m}\geq 0$, and $\{e_{l,m}\}_m$ is an orthonormal basis of $\mathcal H_l$.
Then
\begin{align*}
\rho_{l}^{\rm out}(t)&=\sum_{m} \Lambda_{l,m} {\rm tr}_{\hi_l}[|V_l (t)e_{l,m}\rangle\langle V_{l}(t)e_{l,m}|]=\sum_{m,m'} \Lambda_{l,m} |\psi_{l,mm'}(t) \rangle\langle \psi_{l,mm'}(t)|,
\end{align*}
where $\psi_{l,mm'}(t)\in \mathcal F$ is the unique vector satisfying $\langle \chi|\psi_{l,mm'}(t)\rangle = \langle e_{l,m}\otimes \chi|V_l(t)e_{l,m'}\rangle$ for all $\chi\in \mathcal F$.
Now
\begin{align*}
\langle \psi_{l,mn}(t)|\psi_{l',m'n'}(t)\rangle &=\langle e_{l,n}|T_{ll',t}(|e_{l,m}\rangle\langle e_{l'm'}|)|e_{l',n'}\rangle,
\end{align*}
so we can write
\begin{align*}
{\rm tr}[\rho_{l}^{\rm out}(t)\rho_{l'}^{\rm out} (t)]&= \sum_{n,m}\sum_{n',m'} \Lambda_{l,n}\Lambda_{l',n'} |\langle\psi_{l,mn}(t)|\psi_{l',mn}(t)\rangle|^2\\
&=
\sum_{n,m}\sum_{n',m'}\Lambda_{l,n}\Lambda_{l',n'} |\langle e_{l,n}|T_{ll',t}(|e_{l,m}\rangle\langle e_{l'm'}|)|e_{l',n'}\rangle|^2\rightarrow \delta_{ll'} \sum_{n,n'}\Lambda_{l,n}^2\Lambda_{l,n'}^2.
\end{align*}
\end{proof}

We can now proceed with the proof of Theorem \ref{equiv_thm}. The `if' part is straightforward. Assume now that parameter sets $(H^l,\{L_i^l\}_{i=1}^{k})$, $l=1,2$ are equivalent, and define $T_{ll'}$ as in Lemma \ref{firstlemma}. We consider the direct sum isometry
$$V_{\rm tot}(t):=V_1(t)\oplus V_2(t):\hi_1\oplus\hi_2\to \hi_1\otimes\mathcal F\oplus\hi_2\otimes\mathcal F=(\hi_1\oplus\hi_2)\otimes \mathcal F.$$
We identify the elements $X\in\mathcal B(\hi_1\oplus\hi_2)$ in the usual way with block matrices
$$
X=\begin{pmatrix} X_{11} & X_{12}\\ X_{21} & X_{22}\end{pmatrix},
$$
where $X_{ll'}\in \mathcal B(\hi_{l'},\hi_{l})$, the set of linear operators $\hi_{l'}\to\hi_{l}$. This identifies $\mathcal B(\hi_l,\hi_{l'})$ as a subspace $\mathcal B(\hi_1\oplus\hi_2)$, and each of these four subspaces is invariant under the channels $T_t$ associated with $V_{\rm tot}(t)$. Explicitly, we have
\begin{equation}\label{phirep}
T_t(X) =\begin{pmatrix} T_{11,t}(X_{11}) & T_{12,t}(X_{12})\\ T_{21,t}(X_{21}) & T_{22,t}(X_{22})\end{pmatrix}.
\end{equation}
In particular, any eigenvalue of $T_{ll',t}$ is also an eigenvalue of $T_t$, because the subspaces are invariant. Since each $T_t$ is completely positive and unital by construction, all eigenvalues of $T_{ll',t}$ have modulus at most one, hence the eigenvalues of $\mathbb W_{12}$ have real part $\leq 0$. If all of these are strictly negative, then we have $\lim_{t\rightarrow\infty} T_{12}(t)=\lim_{t\rightarrow\infty}e^{t \mathbb W_{12}}=0$, which according to Lemma \ref{secondlemma} contradicts the assumption that the output states are equal. Hence $\mathbb W_{12}$ must have a purely imaginary eigenvalue, so Lemma \ref{firstlemma} concludes the proof. $\qed$

\subsection{Proof of Proposition \ref{continuous_mc}}

Clearly we may assume $\mathcal C_{\mathsf D}(X^0)=X^0$ (that is, $X^0\in \mathcal B_0$) without loss of generality. The key ingredient is the following lemma which we prove first.
\begin{lemma}\label{iteration} For any tuple of operators ${\bf X}:= (X^0,X^1,\dots,X^{k})\in \mathcal B_0\otimes M(\mathbb{C}^d)^{k}$, 
and all $s\geq 0$, the following equality between maps on $M(\mathbb{C}^d)$ holds
\begin{equation}\label{condexp}
 \sqrt{s} \left\langle\mathbb{F}_s({\bf X}) \Omega \left| j_s(\cdot) \Omega\right\rangle\right.=\int_0^s T_{t}\circ\Phi_{\bf X}\circ T_{s-t} 
(\cdot)\,dt,
\end{equation}
where $\Phi_{\bf X} : M(\mathbb{C}^d)\to M(\mathbb{C}^d)$ is the map given by
$$
\Phi_{\bf X}(Y)=X^{0*} Y+\sum_{i=1}^k X^{i*}[Y ,L^i].
$$
\end{lemma}
\begin{proof}
We let $F_s$ denote the left hand side of \eqref{condexp}. By applying the quantum Ito formula \eqref{itoformula} to the product of two adapted processes inside the conditional expectation, and eliminating the terms involving the annihilation processes acting on the vacuum, we get
\begin{align*}
d F_s(B)&=
\langle \Omega | \left( \sum_{i=1}^{k}j_{s}(X^{i*}) dA_{i}(s) +j_s(X^{0*})ds\right) dj_s(B) \, |\Omega\rangle+
\langle \Omega |  j_s(X^{0*})j_s(B) |\Omega\rangle ds\\
&+ \langle \Omega| \left( \sum_{i=1}^{k}j_s(X^{i*})dA_{i}(s) dj_s(B) \right) \, |\Omega\rangle
\end{align*}
for all system operators $B$. Using now the Langevin equation \eqref{langevin}, together with the Ito multiplication rules, and again eliminating the contributions from the annihilation processes, we get
\begin{align*}
F_s(B)&=\int_0^s F_t(\mathbb W(B))\,dt + \langle \Omega| \int_0^s\left( j_t(X^{0*})j_t(B)+
\sum_{i=1}^{k}j_t(X^{i*})j_t([B,L^i])\right) \, |\Omega \rangle dt \\
&=\int_0^s \left(F_t\circ \mathbb W(B) + T_t\circ \Phi_{\bf X}(B)\right) dt,
\end{align*}
where we have also used \eqref{channel}. Hence $F_s$ satisfies the (ordinary) differential equation
$$
\frac{dF_t}{dt} =F_t\circ \mathbb W +T_t\circ \Phi_{\bf X},
$$
with initial condition $F_0=0$. We can easily solve this equation: without the inhomogeneous part $T_t\circ \Phi_{\bf X}$, the solution would be simply $T_t=e^{t\mathbb W}$; hence the actual solution is obtained by concatenating $T_t\circ \Phi_{\bf X}$ with $T_{s-t}$, and integrating. This gives the claimed result.
\end{proof}

Note that the covariance of the fluctuation operators is sesquilinear with 
respect to the operator coefficients ${\bf X}$ and ${\bf Y}$. We show that the limit exists by computing it explicitly using the Ito calculus. 
The differential of the product $\mathbb{F}_s(\mathbf X)^*\mathbb F_s(\mathbf Y)$ is given by the quantum Ito formula \eqref{itoformula}:
\begin{equation}\label{itoformula2}
d(\mathbb{F}_s(\mathbf X)^*\mathbb F_s(\mathbf Y)) = \mathbb{F}_s(\mathbf X)^* \cdot d \mathbb F_s(\mathbf Y)+ d \mathbb{F}_s(\mathbf X)^* \cdot \mathbb F_s(\mathbf Y)+  d \mathbb{F}_s(\mathbf X)^* \cdot d\mathbb F_s(\mathbf Y)
\end{equation}
For the last term, the Ito rule gives
$ d \mathbb{F}_s(\mathbf X)^* \cdot d\mathbb F_s(\mathbf Y) = \frac{1}{t}\sum_{i=1}^{k}j_s(X^{i*}Y^i) ds$,
and hence by using \eqref{channel} and \eqref{limitchannel}, we get
$$
\int_0^t\langle \varphi\otimes \Omega| d \mathbb{F}_s(\mathbf X)^* \cdot  d\mathbb F_s(\mathbf Y)  |\varphi\otimes \Omega\rangle = 
\frac{1}{t}\int_0^t \left\langle \varphi \left|T_s\left( \sum_{i=1}^{k}X^{i*}Y^i\right)\right|\varphi\right\rangle ds \stackrel{t\rightarrow\infty}{\longrightarrow} \mathrm{tr}[\rho_{ss}\sum_{i=1}^{k}X^{i*}Y^i].
$$
The expectation of the first term in \eqref{itoformula2} can be computed by applying Lemma \ref{iteration} with 
$\Phi_{\bf X}:=iX^{0*}(\cdot)+\sum_i X^{i*}[(\cdot),L_i]$; we get
\begin{align*}
\langle \varphi\otimes \Omega| \int_0^t  \mathbb{F}_s(\mathbf X)^* \cdot d \mathbb F_s(\mathbf Y) |\varphi\otimes \Omega\rangle 
&= \frac{1}{\sqrt{t}}\langle \varphi\otimes \Omega| \int_0^t \mathbb{F}_s(\mathbf X)^* j_s(-iY^0) \,ds |\varphi\otimes \Omega\rangle\\
&= \frac{1}{t}\langle \varphi |\int_0^t ds \int_0^s d r \,T_r \circ \Phi\circ T_{s-r}(-iY^0) \,|\varphi\rangle\\
&= \int_0^t ds \,\langle \varphi | \frac{1}{t}\left(\int_0^{t-s} dr\,T_r\right)\circ \Phi\circ T_{s}(-iY^0) \,|\varphi\rangle\stackrel{t\rightarrow\infty}{\longrightarrow} -\mathrm{tr}[\rho_{ss}\Phi\circ \mathbb W^{-1}(-iY^0)],
\end{align*}
where we have also used the limit relations \eqref{limitchannel} and \eqref{inverselimit}. The second term in \eqref{itoformula2} is obtained by taking the adjoint of the first term with the roles of ${\bf X}$ and ${\bf Y}$ interchanged; this gives
\begin{align*}
(\mathbf{X},\mathbf Y)_{\mathsf{D}}&={\rm tr}\left[\rho_{ss}\left(\sum_i X^{i*}Y^i+\Phi_{\bf X}\circ \mathbb W^{-1}(iY^0)+(\Phi_{\bf Y}\circ\mathbb W^{-1}(iY^0))^*\right)\right]\\
&={\rm tr}\left[\rho_{ss}\left(\sum_{i=1}^{k}X^{i*}Y^i - X^{0*}\mathbb W^{-1}(Y^0)-\mathbb W^{-1}(X^{0*})Y^{0} -i\sum_{i=1}^{k}X^{i*}[L^i,\mathbb W^{-1}(Y^0)]+i\sum_{i=1}^{k}[\mathbb W^{-1}(X^{0*}),L^{i*}]Y^i\right)\right]
\end{align*}
We then apply the identity
$$
\mathbb W(X^*Y)- X^*\mathbb W(Y)-\mathbb W(X^*)Y=\sum_i [L^i,X]^*[L^i,Y],
$$
which holds for arbitrary matrices $X,Y$, to the case where $X=\mathbb W^{-1}(X^0)$ and $Y=\mathbb W^{-1}(Y^0)$. Using the fact  that ${\rm tr}[\rho_{ss}\mathbb W(X^*Y)]=0$, we obtain the following formula for the inner product
\begin{equation}\label{eq.formula}
(\mathbf{X},\mathbf Y)_{\mathsf{D}} = 
\sum_{i=1}^k{\rm tr}\left[\rho_{ss}\left( X^i - i[L^i,\mathbb W^{-1}(X^0)] \right)^*  \left( X^i - i[L^i,\mathbb W^{-1}(X^0)] \right) \right]
\end{equation}
This proves the proposition.

\subsection{Proof of Theorem \ref{th.lan}} Since $\dot{\mathsf{D}}_a\in \mathcal{T}_{\mathsf{D}_0}^{id}$ we have  $\mathcal{E}_{\mathsf{D}_0} (\dot{\mathsf{D}}_a)=0$ and using this we find that the first order term is given by
\begin{equation}\label{eq.l1}
\mathbb L_1[u,u'](X):= \sum_{a,i} u_a  \dot{L}^{i*}_a [X,L^i] - \sum_{a,i} u'_a [X, L^{i*}]  \dot{L}^{i}_a .  
\end{equation}
The second order term is given by
\begin{align*}
\mathbb L_2[u,u'](X) &=\sum_{a, a^\prime}
u_{a} u_{a'} \left(\left(i \ddot{H}_{aa'}-\frac 12\sum_i (\ddot{L}^{i*}_{aa'} L^i+L^{i*}  \ddot{L}^i_{aa'} +
2 \dot{L}^{i*}_{a} \dot{L}^{i}_{a'} ) \right)X +\sum_i  \ddot{L}^{i*}_{aa'} X L^i \right)\\
&+\sum_{a,a'}u'_{a}u'_{a'} \left(X\left(-i\ddot{H}_{aa'}-\frac 12\sum_i (\ddot{L}^{i*}_{aa'} L^i+L^{i*} \ddot{L}^i_{aa'}
+2 \dot{L}^{i*}_{a} \dot{L}^i_{a'})\right) +\sum_i L^{i*} X \ddot{L}^i_{aa'}\right)\\
&+ 2\sum_{aa'}u_{a}u'_{a'}\sum_i  \dot{L}^{i*}_a X\dot{L}^i_{a'}.
\end{align*}
Using a version of Trotter-Kato theorem (cf. Theorem 2.2 in \cite{CatanaGutaBouten}), we obtain the limit
$$
\lim_{t\rightarrow\infty} e^{t \mathbb W_{\mathsf{D}_{u/\sqrt t}, \mathsf{D}_{u'/\sqrt t} }}(\id)=e^{f(u,u')}\id,
$$
where
$$
f(u,u')={\rm tr}\left[\rho^{\mathsf{D}_0}_{ss} \left(\frac 12 \mathbb L_2[u,u'](\id)-\mathbb L_1[u,u']\circ \mathbb W_{\mathsf{D}_0}^{-1}\circ \mathbb L_1[u,u'](\id) \right)\right].
$$

Now, from equation \eqref{eq.l1} we find $\mathbb L_1[u,u^\prime](\id)=0$  and
\begin{align*}
\mathbb L_2[u,u^\prime](\id) &=\sum_{aa'}u_{a}u_{a'} \left(i\ddot{H}_{aa'}+
\frac 12\sum_i (\ddot{L}^{i*}_{aa'} L^i-L^{i*} 
\ddot{L}^i_{aa'}-2\dot{L}^{i*}_a \dot{L}^i_{a'}) \right)\\
&+\sum_{aa'}u^\prime_{a}u^\prime_{a'} \left(-i\ddot{H}_{aa'}+
\frac 12\sum_i (-\ddot{L}^{i*}_{aa'} L^i+L^{i*} \ddot{L}^i_{aa'} -2 \dot{L}^{i*}_a \dot{L}^i_{a'} )\right)\\
&-\sum_{aa'}(u_{a}-u^\prime_a)(u_{a'}-u^\prime_{a'}) 
\sum_i \dot{L}^{i*}_a \dot{L}^i_{a'}
+2i {\rm Im} \sum_{aa'}u_a u^\prime_{a'}\sum_i \dot{L}^{i*}_a \dot{L}^i_{a'}\\
&+\sum_{aa'}(u_au_{a'}+u^\prime_a u^\prime_{a'})\sum_i \dot{L}^{i*}_a \dot{L}^i_{a'}\\
&= \sum_{aa'} \left[- (u_a- u^\prime_a) (u_{a'}- u^\prime_{a'}) \dot{L}^{i*}_{a}\dot{L}^{i}_{a'}
+2i u_a u^\prime_{a'} {\rm Im}  \dot{L}^{i*}_a \dot{L}^i_{a'} +i (u_a u_{a'}-u^\prime_a u^\prime_{a'}) \left(\ddot{H}_{aa'}
+
{\rm Im}\sum_i \ddot{L}^{i*}_{aa'} L^i \right)\right].
\end{align*}
Therefore we have
\begin{eqnarray*}
f(u,u^\prime) &=&\frac{1}{2} {\rm tr} \left(\rho_{ss}^{\mathsf{D}_0}  \mathbb{L}_2 [u,u^\prime] (\id)\right) = 
- \frac{1}{8} (u-u^\prime)^T f^{\mathsf{D}_0} (u-u^\prime) + i u^T \sigma^{\mathsf{D}_0}  u^\prime  + i(u^{T} Su - u^{\prime T} S u^{\prime}).
\end{eqnarray*}
Above, $f^{\mathsf{D}_0}$ is the quantum Fisher information matrix at $\mathsf{D}_0$ whose entries have the simple form due to the fact that the tangent vectors $\dot{\mathsf{D}}_a$ belong to the space $\mathcal{T}^{id}_{\mathsf{D}_0}$ 
$$
f^{\mathsf{D}_0}_{aa'} = 4{\rm Re} ( \dot{\mathsf{D}}_a , \dot{\mathsf{D}}_{a'})_{\mathsf{D}_0} = 
4{\rm Re} \sum_i {\rm tr}(\rho^{\mathsf{D}_0}_{ss} \dot{L}^{i*}_a \dot{L}^{i}_{a'}).
$$ 
Moreover, $\sigma^{\mathsf{D}_0} $ is the symplectic matrix at $\mathsf{D}_0$ (see Definition \ref{def.ccr})
$$
\sigma^{\mathsf{D}_0}_{aa'} = {\rm Im} ( \dot{\mathsf{D}}_a , \dot{\mathsf{D}}_{a'})_{\mathsf{D}_0}=
{\rm Im}  \sum_i {\rm tr}(\rho^{\mathsf{D}_0}_{ss} \ \dot{L}^{i*}_{a}\dot{L}^{i}_{a'})
$$
and $S$ is the real symmetric matrix
$$
S_{aa'}=\frac 12{\rm tr}\left[\rho^{\mathsf{D}_0}_{ss}\left(\ddot{H}_{aa'}+{\rm Im}\sum_i \ddot{L}^{i*}_{aa'} L^i\right)\right].
$$

In conclusion, the overlaps of the system-output states have the following limit 
\begin{eqnarray*}
\lim_{t\to\infty} 
\left.\left\langle \Psi^{\rm s+o}_{u/\sqrt{t}}(t) \right| \Psi^{\rm s+o}_{u^\prime/\sqrt{t}}(t) \right\rangle
&=&
\exp\left( - \frac{1}{8} (u-u^\prime)^T f^{\mathsf{D}_0} (u-u^\prime) + i u^T \sigma^{\mathsf{D}_0}  u^\prime  + i(u^{T} Su - u^{\prime T} S u^{\prime}) \right) \\
&=& e^{i\phi(u) - i\phi(u^\prime)} \langle u|u^\prime\rangle.
\end{eqnarray*}
where $\phi(i):=u^{T} Su $ is a phase angle, and $|u\rangle = W(u)|\Omega\rangle$ is the coherent state on the CCR algebra $CCR(\mathcal{T}^{id}_{\mathsf{D}_0}, \sigma^{\mathsf{D}_0})$ introduced in Definition \ref{def.ccr}, so that the overlaps of two coherent states is given by equation \eqref{eq.overlap.coherent}.


\end{document}